\documentclass[acmsmall, nonacm]{acmart}

\usepackage{mathpartir}
\usepackage{cleveref}
\usepackage{wrapfig}
\usepackage{tikzit}
\usepackage{tikz-cd}
\usepackage{tikz}
\usetikzlibrary{positioning,decorations.markings,arrows.meta,calc,fit,quotes,cd,math,arrows,backgrounds,shapes.geometric}
\usepackage{mathtools,stackengine}
\usepackage{tcolorbox}
\usepackage{suffix}

\tikzstyle{action}=[fill=white, draw=none, shape=circle, tikzit draw=black, inner sep=0.5pt]
\tikzstyle{medium box}=[fill=white, draw=none, shape=rectangle, minimum height=0.1cm, minimum width=0.1cm, tikzit draw=black, inner sep=0.5pt]
\tikzstyle{var}=[fill=white, draw=black, shape=circle, tikzit draw=blue]
\tikzstyle{input}=[fill=white, draw=black, shape=rectangle, minimum size=0.1cm, tikzit draw=red, inner sep=1pt]

\tikzstyle{visible}=[-, fill=none, draw=black, densely dotted, tikzit draw=blue]

\newcommand{\MCC}{\ensuremath{\mathcal C}}

\newcommand{\MCL}{\ensuremath{\mathcal L}}

\newcommand{\MCO}{\ensuremath{\mathcal O}}
\newcommand{\MCP}{\ensuremath{\mathcal P}}

\newcommand{\MCT}{\ensuremath{\mathcal T}}

\newcommand{\MCX}{\ensuremath{\mathcal X}}
\newcommand{\MCY}{\ensuremath{\mathcal Y}}
\newcommand{\MCZ}{\ensuremath{\mathcal Z}}

\newcommand{\fork}{\mathsf{fork}}

\newcommand{\sop}{\mathsf{stop}}
\newcommand{\wait}{\mathsf{wait}}
\newcommand{\print}[1]{\mathsf{act}_{#1}}
\newcommand{\act}[1]{\print{#1}}
\newcommand{\perform}[1]{{\mathsf{perform}}_{#1}}

\newcommand{\gfork}{\mathsf{\underline{fork}}}
\newcommand{\gsop}{\mathsf{\underline{sto}p}}
\newcommand{\gwait}{\mathsf{\underline{wait}}}
\newcommand{\gprint}[1]{{\mathsf{p\underline{erform}}}_{#1}}
\newcommand{\gprintstop}[1]{\mathsf{\underline{act}}_{#1}}

\newcommand{\sfork}{\fork}
\newcommand{\swait}{\wait}
\newcommand{\ssop}{\sop}
\newcommand{\sprint}[1]{\print{#1}}

\newcommand{\vbar}{\mathrel{|}}
\newcommand{\sdots}{\mathinner{{\ldotp}{\ldotp}{\ldotp}}} 
\newcommand{\cterm}[1]{\mathrel{\vdash^{\mathbf{c}}_{#1}}}
\newcommand{\vterm}[1]{\mathrel{\vdash^{\mathbf{v}}_{#1}}}

\newcommand{\tzero}{\mathsf{0}}
\newcommand{\tone}{\mathsf{1}}
\newcommand{\tprod}[3]{\textstyle{\prod_{#1=#2}^#3}}
\newcommand{\tsum}[3]{\textstyle{\sum_{#1=#2}^#3}}
\newcommand{\letin}[3]{\mathsf{let}\, #1=#2\, \mathsf{in}\, #3}
\newcommand{\inj}[2]{\mathsf{inj}_{#1}\,#2}
\newcommand{\ret}[1]{\mathsf{return}\,#1}
\newcommand{\proj}[2]{\mathsf{proj}_{#1}\,#2}
\newcommand{\casesum}[6]{\mathsf{case}\, #1\, \mathsf{of}\, \{\mathsf{inj}_{#4}( #2) \Rightarrow #3\}_{#4=#5}^{#6}}
\newcommand{\lbd}[2]{\lambda #1.\,#2}
\newcommand{\emptid}{\mathtt{0}} 
\newcommand{\app}[2]{#1\,#2}
\newcommand{\casetwo}[5]{\mathsf{case}\, #1\, \mathsf{of}\, \{\mathsf{inj}_1( #2) \Rightarrow #3,\ \mathsf{inj}_2( #4) \Rightarrow #5 \}}
\newcommand{\casevert}[3][]{\mathsf{case}\, #2\, \mathsf{of}\,#1
  \left\{#3\right\}
}
\newcommand{\clause}[4][,]{%
  \mathsf{inj}_{#2}(#3) \Rightarrow #4#1%
}

\newcommand{\tid}{\mathsf{tid}}

\newcommand{\den}[1]{\llbracket #1\rrbracket}
\newcommand{\abs}[1]{\lvert #1 \rvert}

\newcommand{\possop}{\mathsf{s}} 
\newcommand{\repr}[1]{T_{#1}}
\newcommand{\indexcat}{\mathbf{FinRel}}
\newcommand{\Set}{\mathbf{Set}}
\newcommand{\algsig}{\MCO}
\newcommand{\nfset}[1]{\mathsf{NF}_{#1}}

\newcommand{\defeq}{\stackrel{\textrm{def}}=}

\newcommand{\interp}{\textbf{interp}}
\newcommand{\reify}{\textbf{reify}}
\newcommand{\inc}{\textbf{inc}}

\newcommand{\actth}{\MCC}
\newcommand{\forkth}{\MCT}

\newcommand{\tunit}{\mathsf{unit}}
\newcommand{\tidset}{\mathsf{tid\,set}}
\newcommand{\tempty}{\mathsf{empty}}
\newcommand{\maybe}[1]{#1\,\mathsf{option}}
\newcommand{\gnode}[1]{\underline{\mathsf{node}}_{#1}}
\newcommand{\osome}{\mathsf{Some}\,}
\newcommand{\onone}{\mathsf{None}}
\newcommand{\casesnnl}[4]{\mathsf{case\,}#1\mathsf{\,of\,}\{&\mathsf{Some}\,#2\Rightarrow #3;\\&
  \mathsf{None}\Rightarrow #4\}}

\newcommand{\FinRelIntro}{\mathbf{FinRel}}

\newcommand{\Tids}{\mathbb{T}}

\newcommand{\config}[3]{\langle #1;#2;#3\rangle}
\newcommand{\sconfig}[1]{\langle[a]#1\rangle}
\newcommand{\threads}{\mathit{thread}}
\newcommand{\stopped}{\mathit{finished}}
\newcommand{\waiter}{\prec}
\newcommand{\gettids}{\mathit{tids}}

\newcommand{\ar}{\mathsf{ar}}

\newcommand{\symb}{\Sigma}
\newcommand{\node}{\mathsf{node}}
\newcommand{\posetend}{\mathsf{end}}

\newcommand{\id}{\mathsf{id}}
\newcommand{\noderepm}[1]{{S_{#1}}}
\newcommand{\noderep}[1]{{\mathbf{S}_{#1}}}

\newcommand{\FinRel}{\mathbf{FinRel}}

\newcommand{\ctx}[1]{\mathcal{C}[#1]}

\newcommand{\ctxeq}{\stackrel{\mathrm{ctx}}=}

\newcommand{\bl}{\begin{array}[t]{@{}l@{}}}
\newcommand{\el}{\end{array}}

\newcommand\Konig{K\H{o}nig}

\usepackage{tikz-cd}

\setcopyright{cc}
\setcctype{by}
\acmDOI{10.1145/3776706}
\acmYear{2026}
\acmJournal{PACMPL}
\acmVolume{10}
\acmNumber{POPL}
\acmArticle{64}
\acmMonth{1}
\received{2025-07-10}
\received[accepted]{2025-11-06}

\begin{document}

\title[An Equational Axiomatization of Dynamic Threads via Algebraic Effects]{An Equational Axiomatization of Dynamic Threads\\ via Algebraic Effects}
\subtitle{Presheaves on Finite Relations, Labelled Posets, and Parameterized Algebraic Theories}

\author{Ohad Kammar}
\orcid{0000-0002-2071-0929}
\affiliation{%
  \institution{University of Edinburgh}
  \city{Edinburgh}
  \country{United Kingdom}
}
\email{ohad.kammar@ed.ac.uk}

\author{Jack Liell-Cock}
\orcid{0009-0005-7121-8095}
\affiliation{%
  \institution{University of Oxford}
  \city{Oxford}
  \country{United Kingdom}
}
\email{jack.liell-cock@cs.ox.ac.uk}

\author{Sam Lindley}
\orcid{0000-0002-1360-4714}
\affiliation{%
  \institution{University of Edinburgh}
  \city{Edinburgh}
  \country{United Kingdom}
}
\email{Sam.Lindley@ed.ac.uk}

\author{Cristina Matache}
\orcid{0009-0003-6036-6426}
\affiliation{%
  \institution{University of Edinburgh}
  \city{Edinburgh}
  \country{United Kingdom}
}
\email{cristina.matache@ed.ac.uk}

\author{Sam Staton}
\orcid{0000-0002-7149-3805}
\affiliation{%
  \institution{University of Oxford}
  \city{Oxford}
  \country{United Kingdom}
}
\email{sam.staton@cs.ox.ac.uk}

\begin{abstract}
  We use the theory of algebraic effects to give a complete equational axiomatization for dynamic threads.
  Our method is based on parameterized algebraic theories, which give a concrete syntax for strong monads on functor categories,
  and are a convenient framework for names and binding.

  Our programs are built from the key primitives `fork' and `wait'. `Fork' creates a child thread and passes its name (thread ID) to the parent thread.
  `Wait' allows us to wait for given child threads to finish.
   We provide a parameterized algebraic theory built from fork and wait, together with basic atomic actions and laws such as associativity of `fork'.

   Our equational axiomatization is complete in two senses.
   First, for closed expressions, it completely captures equality of labelled posets (pomsets), an established model of concurrency: model complete. Second, any two open expressions are provably equal if they are equal under all closing substitutions: syntactically complete.

  The benefit of algebraic effects is that the semantic analysis can focus on the algebraic operations of fork and wait.
  We then extend the analysis to a simple concurrent programming language by giving operational and denotational semantics.
  The denotational semantics is built using the methods of parameterized algebraic theories and we show that it is sound, adequate, and fully abstract at first order for labelled-poset observations.
\end{abstract}

  


\maketitle
\section{Introduction}
\label{sec:intro}
The theory of algebraic effects provides a way of analyzing semantic aspects of different computational effects in isolation, and separately from other aspects of programming languages, via the algebraic theories from universal algebra.
This paper provides an analysis of concurrency using the methods of algebraic effects.

A theory of algebraic effects for concurrency has proved elusive~\cite{DBLP:books/daglib/p/GlabbeekP10,plotkin-conc-alg-eff}. This is in spite of the success of equational and compositional reasoning in process algebra~\cite{process-algebra,sangiorgi2001pi,DBLP:books/daglib/0067019,DBLP:journals/scp/HoareS14}, and equational theories of concurrency such as concurrent Kleene algebra~\cite{gischer,cka}. Even more paradoxically, algebraic effects have already inspired powerful concurrency libraries~\cite{SivaramakrishnanDWKJM21,PhippsCostinRGLHSPL23}, but these software implementations do not yet tie with the theories of algebraic effects in terms of universal algebra and category theory. (See~\S\ref{sec:conc-etc} for further discussion of the literature.)

The key technique in our work is \emph{to take thread IDs seriously}. This necessitates an algebraic framework that supports abstract names or IDs, and binding and passing them. For this, we use `parameterized algebraic theories'~\cite{Staton13,Staton13PL}, which already have a tight connection with monads and algebraic effects. There are four operations in our algebraic theory:
\begin{itemize}
    \item $\fork$: Forking a child thread. This is the key operation and is written
$\fork(a.x(a),y)$.
This spawns a new child thread with ID~$a$, running continuation $y$, while concurrently running continuation $x$ in the parent thread, which is passed the ID~$a$ of the child.
\item $\wait$: A command to wait for a thread to end before proceeding.
\item $\sop$: A command to end the current thread now.
  Invoking this command will unblock all the threads waiting for the current one.
\item $\act\sigma$: Primitive atomic actions. Aside: going forward, we could combine with other algebraic effects, such as memory access to look at concurrent shared memory, but for now to focus on concurrency we restrict attention to primitive atomic actions.
\end{itemize}

\paragraph{Contributions}We present a theory with eight equations between these four operations (\S\ref{sec:pres-theory-fork}). We give a syntax-free representation theorem of terms modulo equations (\S\ref{sec:rep-theorem}), and show that for closed terms, the representation exactly matches the long-established model of true concurrency based on labelled posets (`pomsets', Thm.~\ref{thm:lab-poset-free-model}). In fact, this might be the first basic syntactic theory for labelled posets.
For open terms with free variables, we prove a completeness theorem: there can be no further equations on open terms while retaining the labelled posets model on closed terms (\S\ref{sec:compl-theor-theory}).

Algebraic effects allow us to focus on a particular theory, without worrying about other programming language primitives, but it is typically easy to return to a fuller programming language having analyzed the algebraic effects.
In Section~\ref{sec:background}, we give a typical functional programming language with concurrency primitives and an operational semantics. In Section~\ref{sec:denot-semant-progr} we use the algebraic effects and the representation theorem to build a denotational semantics for the programming language that is sound, adequate, and fully abstract at first order.

\subsection{Motivating Fork and Wait with Thread IDs as Language Primitives}

In the previous section we introduced an operation $\fork(a.x(a),y)$, where $\fork$ has type $(\tid \to \mathtt{t}) \to \mathtt{t} \to \mathtt{t}$ polymorphic in $\mathtt{t}$, and $\tid$ is the type of thread IDs.
The style of programming with operations such as $\fork$ is unusual, but according to the theory of algebraic effects, algebraic operations have a counterpart in generic effects~\cite{pp-algop-geneff}, and this matches more closely to realistic languages with effects.
The generic effect for `$\fork$' is a command $\gfork:\tunit\to (\maybe \tid) $,
\[
  \gfork() = \fork(a.\ret (\osome a),\ret \onone)\text.
\]
The algebraic operation $\fork$ can be recovered from the command $\gfork$ by pattern matching on the result of $\gfork$ and using explicit sequencing.

We provide a mini-programming language with an operational semantics in \Cref{sec:background}, which works with pools of threads.
There, $\gfork$ will spawn a new child thread into the thread pool, and the continuation is duplicated. The caller of $\gfork$ can check whether they are the parent or child by
looking at the return value of $\gfork$, and if they are the parent they will be given the
ID of the child, otherwise $\onone$.
Indeed this generic operation $\gfork$ is reminiscent of the POSIX \texttt{fork} construct~\cite{ieee1003.1-2024},
which in the POSIX standard is typed \texttt{pid\_t fork(void)}, which returns the child ID to the parent and~$0$ to the child.
%

Alongside the standard programming primitives, our other generic effects, which match the algebraic operations $\wait$, $\act{\sigma}$, $\sop$, are:
\[
  \gwait:\tid\to \tunit\qquad \gprint\sigma:\tunit\to \tunit \qquad\gsop:\tunit \to \tempty\text.
\]
Here: $\gwait(a)$ puts the current thread into a waiting state, recording for the scheduler the thread~$a$ that it is waiting for; $\gprint\sigma()$ performs the action $\sigma$ immediately, which is recorded by a label in our transition system; and $\gsop()$ ends the current thread, unblocking all other threads that were waiting for it. Here, the type $\tempty$ shows that nothing else will happen on this execution path.

Our operational semantics uses a labelled transition system that records the actions performed.
Inspired by true concurrency models such as asynchronous transition systems (e.g.~\cite{mn-async-ccs}),
we also include some location information, by way of noting the ID of the thread that performed each action. In this simple situation, this is sufficient to observe not only the traces of actions but also the independence between different actions. We can thus, from the operational semantics, obtain a labelled partial order, labelled by actions $\sigma$. Labelled posets, sometimes called `pomsets', are another model of `true' concurrency~\cite{DBLP:journals/ijpp/Pratt86}. For our semantics, the linearizations of the posets are exactly the execution traces of the program.

By defining `well-formed configurations' for our particularly simple language,
we can show that programs never deadlock, roughly because a child can never wait for its parent.
We also show that every closed program determines a unique labelled poset.
This clarifies that our language is very simple, in that programs all terminate, and there is no `conflict' in the sense of event structures~\cite{event-structures}, nor are there any `races'.
These are useful properties to have, and also useful for later relating to denotational semantics, but they do imply that we are studying a very idealized situation compared to how dynamic threads work in practice.
We expect future work to extend with other primitives that allow recursion and conflict.

\subsection{A Simple Complete Fragment for Labelled Posets (Pomsets)}\label{sec:simple-compl-fragm}
This simple language allows us to construct all labelled posets, in other words, it completely describes that model of true concurrency.
To show this, we define $\gnode\sigma:\tidset\to \tid$ (whose type signature is stated informally for now) by
\begin{align*}
  \gnode\sigma([a_1,\dots, a_n])\quad\defeq\quad \\
  \casesnnl {\gfork()} {b}{\ret b}{\gwait(a_1);\dots;\gwait(a_n);\gprint\sigma(); \mathsf{case\,}\gsop\mathsf{\,of\,}\{\} }\text.
\end{align*}
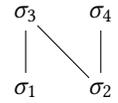
\begin{wrapfigure}[5]{r}{0.15\textwidth}\vspace{-6mm}
\ \ \begin{tikzpicture}
	\begin{pgfonlayer}{nodelayer}
		\node [style=action] (0) at (-1, -1) {$\sigma_1$};
		\node [style=action] (1) at (1, -1) {$\sigma_2$};
		\node [style=action] (2) at (-1, 1) {$\sigma_3$};
		\node [style=action] (3) at (1, 1) {$\sigma_4$};
	\end{pgfonlayer}
	\begin{pgfonlayer}{edgelayer}
		\draw (0) to (2);
		\draw (1) to (2);
		\draw (1) to (3);
	\end{pgfonlayer}
\end{tikzpicture}\caption{N-shape poset\label{fig:N}}
  \end{wrapfigure}
 In the language fragment containing $\gnode\sigma$ and $\gsop$, but without $\gfork$ and $\gwait$, every thread ID performs exactly one action.
Thus the induced labelled poset is a partial order on thread IDs, recording which waits for which, each labelled with their action.
The command $\gnode\sigma([a_1,\dots, a_n])$ adds a node labelled~$\sigma$ to the labelled poset, setting its immediate predecessors to $a_1,\dots, a_n$, and returns the name of the new node.

  For example, the following program induces the N-shape poset (Fig.~\ref{fig:N}).
\[
  \letin {a_1}{\gnode{\sigma_1}([])}
  \letin {a_2}{\gnode{\sigma_2}([])}
  \letin {a_3}{\gnode{\sigma_3}([a_1,a_2])}
  \letin {a_4}{\gnode{\sigma_4}([a_2])} \gsop\]

We can completely axiomatize labelled posets by two axioms, written informally for now using the $(\vec{\ \ })$ and $[\ ]$ notation to denote sets of thread IDs.
See~\Cref{exa:act-theory} for the formal statement.
\begin{align*}
  \letin c{\gnode{\sigma_1}(\vec a)}
  {\letin d{\gnode{\sigma_2}(\vec b)}{[c,d]}}
    \quad&=\quad
    \letin d{\gnode{\sigma_2}(\vec b)}
  {
    \letin c{\gnode{\sigma_1}(\vec a)}
    {[c,d]}}
  \\
   \letin b {\gnode\sigma(\vec a)}
    {[b]}
    \quad &=\quad
    \letin b {\gnode\sigma(\vec a)}
    {[b] \, {+\!\!+} \, \vec a}
  \end{align*}
  The first law says that it does not matter in which order we add nodes, as long as ID dependencies are respected, and the second captures the transitivity of the partial order.
  We show that these two axioms are complete in Theorem~\ref{thm:lab-poset-free-model}, using the more formal framework of parameterized algebraic theories.
  The key point is that by passing around the thread IDs, we are able to fully describe the established model of true concurrency based on labelled posets.

\subsection{Technical Setting: Functor Categories and Naturality for Syntax with Binding and Semantics with Names}

To cope with the dynamic threads and varying number of thread names, we follow the long-standing tradition of using functor categories. In brief, let $\FinRelIntro$ be the category of finite sets of names and relations between them, and let $\Set$ be the category of all sets and functions.
Then the computations at some type $A$ form a functor $\den A:\FinRelIntro\to \Set$,
mapping a set $w$ of available thread IDs to the set $\den A(w)$ of computations that use at most those thread IDs.

According to the functorial action here, for each relation $R\subseteq w\times w'$,
we have a reindexing function $\den A(w)\to \den A(w')$. The idea is that
if a computation in world $w$ would wait for some thread ID~$a\in w$,
then we can transform it into one that instead waits for all the thread IDs in the direct image,
$\{b~|~R(a,b)\}$.

Programs of type $A\to B$ are interpreted as families of functions $\den A(w)\to \den B(w)$ that are moreover natural.
This, in particular, maintains the invariant that one cannot sum thread IDs, guess thread IDs, or compare them in some order. This is similar to the role of names in nominal sets~\cite{nom}.

The framework of parameterized algebraic theories is an established method for algebraic effects over functor categories and admits a concrete syntax, which we use for our axiomatizations.
Moreover, every parameterized algebraic theory induces a strong monad on the functor category, in our case on $[\FinRelIntro, \Set]$.
Monads on functor categories have long been used for denotational semantics of dynamic allocation~\cite{PlotkinP02,oles}.
In~\Cref{sec:interpretation}, we use the strong monad induced by our theory of dynamic threads to give Moggi-style denotational semantics~\cite{moggi_notions_1991}
to the mini-programming language of~\Cref{sec:background}.

\subsection{Fork and Wait in General, Parallel Composition, and Labelled Posets with Holes}\label{sec:fork-wait-general}

Although the $\gnode\sigma$ effect is enough to build all labelled posets, it is more paradigmatic to allow higher level parallelism through $\gfork$ and $\gwait$. For example, we can define a program that puts two other programs in parallel, $\mathsf{parallel}:((\tunit\to \tempty), (\tunit\to \tempty))\to \tempty$, by spawning two threads and waiting for them:
\begin{displaymath}
\bl
  \mathsf{parallel}(x,y) =
  \bl
  \mathsf{case}\,(\gfork())\,\mathsf{of} \{ \\
  \quad
    \begin{array}[t]{@{}l@{~}c@{~}l@{}}
      \osome a &\Rightarrow&
        \bl
        \mathsf{case}\,(\gfork())\,\mathsf{of}\,\{ \\
        \quad
          \begin{array}[t]{@{}l@{~}c@{~}l@{}}
          \osome b &\Rightarrow& \gwait(a);\gwait(b);\gsop()\\
          \onone   &\Rightarrow& y()\}\\
          \end{array}
        \el \\
      \onone &\Rightarrow& x()\} \\
    \end{array}
  \el \\
  \el
  \qquad\quad
  \raisebox{-30pt}[0pt]{\begin{tikzpicture}
	\begin{pgfonlayer}{nodelayer}
		\node [style=var] (0) at (-1, 0) {$x$};
		\node [style=var] (1) at (1, 0) {$y$};
		\node [style=none] (2) at (0, 2) {};
	\end{pgfonlayer}
	\begin{pgfonlayer}{edgelayer}
		\draw (0) to (2.center);
		\draw (1) to (2.center);
	\end{pgfonlayer}
\end{tikzpicture}
}
\end{displaymath}

For this reason, we provide an equational theory for $\fork$ and $\wait$ in Section~\ref{sec:presentation-theory}. A general idea is that $\fork(a.t,u)$ behaves like a monoid, with $\gwait(a);\gsop()$ like a unit, except care is needed for the thread ID parameter.

The \textbf{main result} of our paper is the representation theorem for the fork/wait theory (Theorem~\ref{thm:main-theorem}). This representation is along the lines of the labelled posets, except now there may be holes standing for the different continuations.
For example, the `$\mathsf{parallel}$' operation becomes the `cherries' diagram shown.
%
%
This is non-trivial because any thread plugged in for $x$ or $y$ may have child threads that are not waited for, and may wait on other thread IDs that are not in the diagram.

We also prove a completeness theorem (Theorem~\ref{thm:completeness}), which says that if two expressions give the same labelled poset whatever we substitute into the variables $x$, $y$ etc., then they are provably equal. We show this by finding special gadgets to substitute for the variables.
This can be thought of as a full abstraction result, and we make this connection in Theorem~\ref{thm:full-abstraction}.

For a final remark, we define an operation
${\mathsf{series}:((\tunit\to \tempty), (\tunit\to \tempty))\to \tempty}$
that forks a child thread and immediately waits for it:
\begin{align*}
  \mathsf{series}(x,y)=
  \mathsf{case}\,\gfork()\,\mathsf{of}\,\{&\osome a\Rightarrow \gwait(a);y()~|~
                                                     \onone \Rightarrow x()\}\end{align*}
    We can use `$\mathsf{series}$' and `$\mathsf{parallel}$' to build series-parallel graphs~\cite{McKee1983,Alur2023}, and
    we can easily deduce from our algebraic theory that the equational laws of series-parallel graphs hold.
    But note that we can also express the N shape, which is not series-parallel.

    Although `$\mathsf{series}$' is easy to program, it is not the same as the sequencing $x();y()$ of the programming language; `$\mathsf{series}$' requires the return type of $x$ and $y$ in to be the same. Another view is that the unit of $(;)$ is $(\ret ())$ (return a value) whereas the unit of $\mathsf{series}$ is $\gsop$ (end the current thread).
    In particular, $\gsop();x()$ does not execute $x$ at all.
    Note also that if a child returns without ever invoking $\gsop$, a thread waiting on this child will never unblock.

    The apparent similarity between `$\mathsf{series}$' and sequencing may be the reason for earlier claims that concurrency via algebraic effects has `undesired equations'~\cite[\S 8,p33]{DBLP:books/daglib/p/GlabbeekP10}, such as parallel composition commuting with sequencing. By focusing instead on $\gfork$, $\gwait$, and dynamic threads, we make clear the distinction between sequencing and series: even though `$\mathsf{parallel}$' commutes with sequencing, it does not commute with `$\mathsf{series}$' as expected.

\paragraph{Paper Summary}
The operational and denotational semantics are given in Sections~\ref{sec:background} and~\ref{sec:denot-semant-progr} respectively; the main programming language results are soundness, adequacy and full abstraction (\S\ref{sec:adeq-cont-equiv}).
The method for building the denotational semantics is via algebraic effects, developed in Sections~\ref{sec:param-algebr-theor}--\ref{sec:pres-theory-fork}, shown to have a representation (\S\ref{sec:rep-theorem}) and thereby to be complete (\S\ref{sec:compl-theor-theory}).

\section{Background Concurrent Programming Language}
\label{sec:background}

To precisely frame the situation, we discuss a fairly standard concurrent programming language and operational semantics.

\subsection{Language and Type System}\label{sec:type-system}

We consider a standard higher-order programming language with finite product and sum types (e.g.~\cite{DBLP:journals/iandc/LevyPT03,moggi_notions_1991}).
The grammar of types is:
\begin{displaymath}
  A,B \Coloneqq \tid \vbar \tprod{i}{1}{k}A_i \vbar \tsum{i}{1}{k} A_i \vbar A\rightarrow B
\end{displaymath}
When $k=0$ we get the empty product and sum types, denoted by $\tone$ and $\tzero$ respectively, instead of $\tunit$ and $\tempty$ in the introduction.
The base type of thread IDs $\tid$ is specific to our setting.

The language is fine-grain call-by-value~\cite{DBLP:journals/iandc/LevyPT03}, meaning that terms are stratified into values and computations.
Values are terms that do not beta reduce and do not perform any effects:
\begin{displaymath}
  v \Coloneqq  x\vbar (v_1,\sdots,v_k) \vbar \inj{i}{v} \vbar \lbd{x}{t} \vbar a
  \vbar \emptid \vbar v_1\oplus v_2 \vbar g
\end{displaymath}
Values include variables, the usual constructors for product and sum types, and functions; the body of a function, $t$, is a computation.
The symbol $a$ ranges over a countably infinite set $\Tids$ of actual thread IDs.
The constant $\emptid$ stands for the empty set of thread IDs and $\oplus$ takes the union of two sets of thread IDs.
The symbol $g$ ranges over a fixed set $\mathbb{F}$ of typed term constants $g:A\rightarrow B$.
These constants allow us to add concurrency features to our languages.

The grammar of computations contains the usual destructors for product, sum and function types, and a $\mathsf{let}$-construct for explicitly sequencing computations:
\begin{displaymath}
  t \Coloneqq \ret{v} \vbar \proj{i}{v} \vbar \casesum{v}{x_i}{t_i}{i}{1}{k} \vbar \app{v_1}{v_2} \vbar \letin{x}{t_1}{t_2}
\end{displaymath}
We include the following set $\mathbb{F}$ of value constants $g:A\rightarrow B$ in our language, where $\sigma$ ranges over a fixed set of observable actions $\Sigma$:
  \begin{equation}\label{eqn:fork-wait-geff-prims}
    \gfork : \tone\rightarrow \tid +\tone \qquad \gsop:\tone\rightarrow \tzero \qquad \gwait : \tid \rightarrow \tone  \qquad \gprint{\sigma}:\tone\rightarrow \tone
  \end{equation}
  Intuitively, $\gfork()$ forks a new child thread and returns its ID to the parent thread, in the left branch of the sum type $\tid+1$.
  The right branch is for the child thread which receives the unit value $()$.
  The continuation of $\gfork()$ will run twice, once in the parent thread and once in the child thread.

  The computation $\gsop()$ signals that the current thread has finished and its continuation will be discarded.
  Once a thread has invoked $\gsop()$ it cannot resume running.
  The computation $\gwait(v)$ waits for all the threads with IDs in $v$ to finish by invoking $\gsop()$, then returns unit.
  Finally, $\gprint{\sigma}()$ performs the observable action $\sigma$ then returns unit.

When writing programs, we use some syntactic sugar, such as $(t_1;t_2)$ for $\letin{x}{t_1}{t_2}$ where $t_2$ does not depend on $x$, and $\mathsf{case}\,t\,\mathsf{of}\dots$ instead of
  $\letin x t \mathsf{case}\,v\,\mathsf{of}\dots$ where it is unambiguous.

\begin{example}\label{ex:prog-fork-1}
  Consider the computations below:
  \begin{align*}
    &\letin{y}{\gfork()}{\casetwo{y}{x_1}{\gwait(x_1);\gprint{\sigma_1}();\gsop()}{}{\gprint{\sigma_2}();\gsop()}}\\&
\letin{y}{\gfork()}{\casetwo{y}{x_1}{\gprint{\sigma_1}(); \gwait(x_1);\gsop()}{}{\gprint{\sigma_2}();\gsop()}}
  \end{align*}
  In both, the main thread forks a new child thread that performs action $\sigma_2$ then stops; the ID of this new thread is bound to $x_1$.
  In the first, the main thread waits for the child to finish before performing action $\sigma_1$.
  So the sequence of observed actions is $\sigma_2$ followed by $\sigma_1$.
  In the second, the main thread does not wait, so we may also see the other order, or $\sigma_1$ and $\sigma_2$ concurrently.
\end{example}

There are separate typing judgements $\vterm{w}$ and $\cterm{w}$ for values and computations respectively.
  The judgements that are standard are shown in \Cref{fig:std-typing-rules}.
  All judgements are annotated with a finite subset $w \subseteq \Tids$ of thread IDs, called a world, that does not change throughout a typing derivation.
  A world $w$ stands for threads that have been created by the environment.
  A term may use a thread ID in the world via the typing rule:
  \begin{displaymath}
    \inferrule
    {a\in w}
    {\Gamma \vterm{w} a:\tid}
  \end{displaymath}



For equational reasoning about programs, it is very helpful to combine thread IDs into \emph{compound} thread IDs.
For example, suppose that a thread~$a$ has nothing left to do but is waiting for~$b$ and~$c$ before it finishes. Then thread~$a$ is rather redundant, and the only reason to keep it is that the parent thread might at some point wait for~$a$. If the parent could instead wait for both~$b$ \emph{and}~$c$, then we can indeed finish~$a$ already. (This is an instance of axiom~(\ref{eq:fork-sub}).)
To reason in this example it is helpful to substitute the compound thread ID $(b\oplus c)$ for $a$.
To facilitate this equational reasoning, which is the aim of this paper, we have the following constructions of compound threads IDs, which in the operational semantics are treated as sets of IDs:
\begin{mathpar}
  \inferrule
  { }
  {\Gamma \vterm{w} \emptid : \tid}
  \and
  \inferrule
  {\Gamma \vterm{w} v_1:\tid \\ \Gamma \vterm{w} v_2:\tid }
  {\Gamma \vterm{w} v_1\oplus v_2:\tid}
  \and
\end{mathpar}

\begin{figure}
  \centering
\begin{mathpar}
  \inferrule
  { }
  {\Gamma,x:A,\Gamma' \vterm{w} x:A}
  \and
  \inferrule
  {\bigl(\Gamma \vterm{w} v_i:A_i\bigr)_{i=1}^k}
  {\Gamma \vterm{w} (v_1,\sdots,v_k): \tprod{i}{1}{k}A_i }
  \and
  \inferrule
  {\Gamma \vterm{w} v_i:A_i}
  {\Gamma \vterm{w} \inj{i}{v_i} : \tsum{i}{1}{k}A_i}
  \and
  \inferrule
  {\Gamma,x:A \cterm{w} t:B}
  {\Gamma \vterm{w} \lbd{x}{t}:A\rightarrow B}
  \and
  \inferrule
  {\Gamma \vterm{w} v:A}
  {\Gamma \cterm{w} \ret{v}:A }
  \and
  \inferrule
  {\Gamma \vterm{w} v:\tprod{i}{1}{k}A_i}
  {\Gamma \cterm{w} \proj{i}{v}:A_i}
  \and
  \inferrule
  {\Gamma \vterm{w} v: \tsum{i}{1}{k}A_i \and \bigl(\Gamma,x_i:A_i \cterm{w} t_i :B \bigr)_{i=1}^k}
  {\Gamma \cterm{w} \casesum{v}{x_i}{t_i}{i}{1}{k} :B}
  \and
  \inferrule
  {\Gamma \vterm{w} v_1:A\rightarrow B \and \Gamma \vterm{w} v_2:A}
  {\Gamma \cterm{w} \app{v_1}{v_2} : B}
  \and
  \inferrule
  {\Gamma \cterm{w} t_1:A \and \Gamma,x:A \cterm{w} t_2:B}
  {\Gamma \cterm{w} \letin{x}{t_1}{t_2}:B}
  \and
  \inferrule
  {(g:A\rightarrow B) \in \mathbb{F} }
  {\Gamma \vterm{w} g : A\rightarrow B}
\end{mathpar}
\caption{Standard typing rules for a fine-grain call-by-value programming language.
  Here ${\mathbb{F}}$ is given in \eqref{eqn:fork-wait-geff-prims}.}\label{fig:std-typing-rules}
\end{figure}


\subsection{Operational Semantics}
\label{sec:opsem}

We now define an operational semantics for the language, based on a labelled transition
relation over configurations. These are pools of threads, some of which are ready to run, some are finished, and some are stuck waiting for others to finish before they can run. We include a relation stating which threads are waiting for which others to finish.

\subsubsection{Alternative Language Construct: $\gprintstop\sigma$} \label{sec:alt-lang}

In order to set up the operational semantics, it is convenient to consider the following
operation
\[
  \gprintstop \sigma : \tone \to \tzero
\]
which performs action $\sigma$ and then finishes immediately.
This is interdefinable with ${\gprint\sigma:\tone\to \tone}$:
\begin{align*}
  \gprintstop \sigma() \ \ &= \ \ \gprint\sigma();\gsop()
\\\text{and}\quad
  \gprint \sigma ()\ \ &=\ \  \letin x {\gfork()} {\casetwo x {a}{\gwait(a)}{}{\gprintstop \sigma()}}\text.
\end{align*}
That is, an action that may be followed by other commands can be achieved by forking a new thread that merely performs the action, and then waiting for it.

Arguably, $\gprint\sigma$ is more natural in programming, but $\gprintstop\sigma$ has an easier semantics. We focus on semantics, and so we focus on $\gprintstop\sigma$ as a primitive, regarding $\gprint\sigma$ as derived.

\subsubsection{Configurations}
\begin{definition}
A \emph{configuration} is a triple $\config w \waiter \threads$
where
\begin{itemize}\item $w\subseteq \Tids$ is a finite set of thread IDs that are involved in this configuration.
\item $(\waiter)\subseteq \Tids\times w$ is a relation, relating thread IDs in $w$ with the, potentially external, IDs from $\Tids$ that they are waiting for.
\item $\threads:w\to \{\text{computations $t$}\}\uplus \stopped$ is a function assigning to each active thread id the computation that it runs, or `$\stopped$' if the thread has finished.
  We often enumerate the map e.g.~writing $([a]t,[b]u)$, if $a\mapsto t$ and $b\mapsto u$.
\end{itemize}
\end{definition}
We abbreviate the configuration when there is one thread,
writing $\sconfig t$ instead of $\config {\{a\}} \emptyset {[a]t}$.

\subsubsection{Transition Relation}
The transition system is given inductively in Figure~\ref{fig:opsem}.
We define two transition relations between configurations: silent reductions $\longrightarrow$ and
labelled reductions $\xlongrightarrow \sigma$, where $\sigma$ is an action.
We also annotate our transition relation with the thread which reduced ($a$), following~\cite{mn-async-ccs}; this is not necessary and can be erased, but is useful in the metatheory.

The relation $a \waiter b$ specifies that thread $b$ is waiting on thread $a$ to finish. The last transition rule in Figure~\ref{fig:opsem} says that
a thread can step if indeed all the threads it was waiting for have finished.
After a step, the waiting relation $(\waiter)$ needs to be updated with any new waits $(\waiter')$.

The other transition rules are for the reduction of single threads.
Wait reduces by recording what it is waiting for.
Fork spawns a new child thread $b$, passing its identifier $b$ to the parent thread~$a$.

In this simple language, there is only one evaluation context, $\letin x {[-]} s$.
To evaluate here depends on which threads are reduced, spawned or finished by the expression in the context ($t$). We then continue to evaluate the let-expression
with all of the threads from $w$, each of which proceeds with its own copy of the continuation $s$; any finished threads will finish without evaluating $s$.

There are a couple of subtle points about the $\waiter$ relation.
First, a configuration $\config w\waiter \threads$ may have $b\waiter a$ for some $b$ not in $w$. Thus a thread may wait on thread IDs not in the current pool. This is to allow us to
restrict our view to particular threads, but will not happen at the top level.

Second, a configuration may have $a\waiter b$ even if both~$a$ and~$b$ are finished.
One could garbage-collect this redundant information, as any efficient implementation would do, but this is not necessary, and the metatheory is easier without it.

\begin{figure}
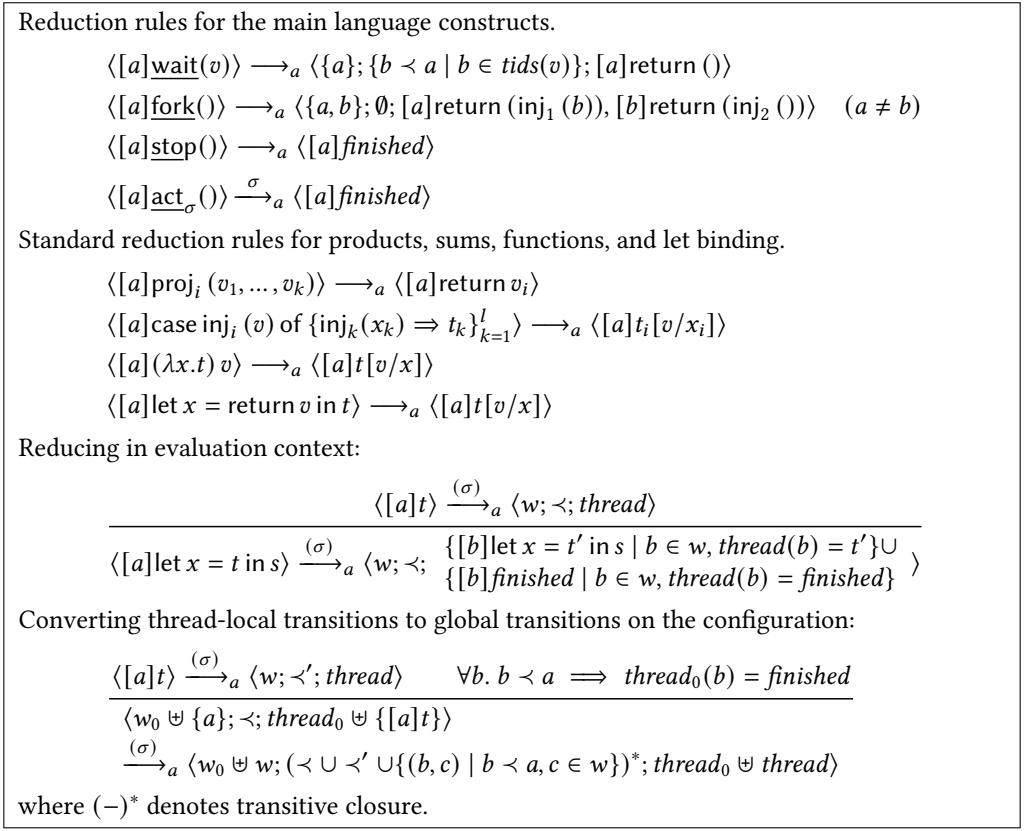

  \framebox{
    \begin{minipage}{0.95\linewidth}Reduction rules for the main language constructs.
 \begin{align*}
   &\sconfig {\gwait(v)}
  \longrightarrow_a
     \config {\{a\}}{\{b\waiter a~|~b\in\gettids(v)\}} {[a]\ret()}
\\&  \sconfig{\gfork()}
  \longrightarrow_a
   \config   {\{a,b\}}\emptyset{[a]\ret(\inj 1(b)),[b]\ret(\inj 2())}
   \quad (a\neq b)
\\   &
     \sconfig{\gsop()}
  \longrightarrow_a
  \sconfig{\stopped}
   \\&
  \sconfig {\gprintstop\sigma()}
  \xlongrightarrow\sigma_a
  \sconfig{\stopped}
\intertext{   Standard reduction rules for products, sums, functions, and let binding. }
&  \sconfig{\proj{i}(v_1,\sdots,v_k) }
  \longrightarrow_a
  \sconfig{\ret{v_i}}
\\&
   \sconfig {\raisebox{0pt}[0pt]{\(\casesum{\inj{i}(v)}{x_k}{t_k}{k}1{l}\)}}
   \longrightarrow_a \sconfig{t_i[v/x_i]}
\\&
   \sconfig{(\lambda x.t)\,v}
   \longrightarrow_a \sconfig{t[v/x]}
   \\&  \sconfig {\letin{x}{\ret{v}}{t}}
   \longrightarrow_a \sconfig {t[v / x]}
\intertext{Reducing in evaluation context:}
&  \inferrule
  { \sconfig{t} \xlongrightarrow{(\sigma)}_a \config w\waiter \threads }
  { \sconfig{\letin{x}{t}{s}}
    \xlongrightarrow{(\sigma)}_a
    \config w  \waiter {\mbox{\(\begin{array}{l}\{[b]\letin{x}{
      t'}{s}~|~b\in w,\threads(b)=t'\}
    \cup \\\{[b]\stopped~|~b\in w,\threads(b)=\stopped\}\end{array}\)}}}
\intertext{Converting thread-local transitions to global transitions on the configuration:}
&  \inferrule
  { \sconfig{t} \xlongrightarrow{(\sigma)}_a
    \config {w}{\waiter'}{\threads}
    \qquad \forall b. \ b\waiter a \implies \threads_0(b)=\stopped}
  {\mbox{\(\begin{array}{l} \config {w_0\uplus \{a\}}\waiter {\threads_0\uplus \{[a]t\}}
\\    \xlongrightarrow{(\sigma)}_a
    \config { w_0\uplus w}
    {(\waiter\cup\waiter'\cup\{(b,c)~|~b\waiter a, c\in w\})^*}
    {\threads_0\uplus \threads}\end{array}\)}}
\end{align*}
where $ (-)^*$ denotes transitive closure.
\end{minipage}}
\caption{Operational semantics for our concurrent programming language (Sec.~\ref{sec:opsem}). We write $\xlongrightarrow {(\sigma)}$ with parentheses to indicate two copies of the rule,
  one with the label and one without.\label{fig:opsem}
  Here $\gettids(a)=\{a\}$, $\gettids(v\oplus v')=\gettids(v)\cup \gettids(v')$, and
     $\gettids(\emptid)=\emptyset$.
}
\end{figure}

\subsubsection{Observation as a Labelled Poset}
We focus on true-concurrency semantics,
and so, instead of considering only linear traces,
we include the dependency order $\waiter$.
This gives a labelled poset (pomset~\cite{DBLP:journals/ijpp/Pratt86,pp-teams}, equivalently conflict-free event structure~\cite{event-structures}).
\begin{definition} \label{def:labelled-poset}
  Let $\Sigma$ be a set. A \emph{$\Sigma$-labelled poset} is a
  partially ordered set $P=(X,\leq)$ equipped with a function $\ell:X\to \Sigma$.
  (We omit $\Sigma$ where it is clear from the context.)
\end{definition}

\begin{definition}\label{def:observ-as-labell}
  A \emph{terminal configuration} $\config w\waiter \threads$ is one where all threads have finished:
  $\threads(a)=\stopped$ for all $a\in w$.

  Let $(X,\leq, \ell)$ be a labelled poset and $C$ a configuration.
  We write
  \begin{displaymath}
    C\Downarrow (X,\leq,\ell)
  \end{displaymath}
  when there is some terminal configuration $C'=\config w\waiter \threads$ such that $C$ admits a sequence of transitions
  \[
    C\longrightarrow^* \xlongrightarrow{\sigma_1}_{a_1}\longrightarrow^*\xlongrightarrow{\sigma_2}_{a_2} \cdots
    \longrightarrow^*\xlongrightarrow{\sigma_n}_{a_n}    \longrightarrow^* C' \text,
  \]
  and the following conditions on $(X,\leq,\ell)$ hold: $X=\{a_1,\dots, a_n\}$, the order on $X$ is given by $a_i\leq a_j$ iff $a_i\waiter a_j$ or $i=j$, and $\ell(a_i) = \sigma_i$.
\end{definition}
Recall from the reduction rules in~\Cref{fig:opsem} that each action $\sigma_i$ happens in a separate thread that finishes immediately.
So the set $X$ consists of all the (distinct) IDs, $a_1,\dots,a_n$, of the threads which act, $\ell$ specifies what each action is, and $\leq$ encodes the causal dependencies between actions.
In Example~\ref{ex:prog-fork-1}, the first program is related by $\Downarrow$ to the order
$(\sigma_2<\sigma_1)$, whereas the second one is related to the discrete order $\{\sigma_1,\sigma_2\}$.
See~\cite{pp-teams} for further discussion of these concurrent notions of observation.

\subsection{Operational Meta-theory}
A
well-typed program will never deadlock, and moreover it induces a unique observed labelled poset. More elaborate languages would not have these properties,
but in this simple setting they are useful for
connecting exactly with the denotational semantics in Section~\ref{sec:denot-semant-progr}.
\subsubsection{Well-Formed Configurations}
Well-typed programs never deadlock, that is, there are never two threads that are waiting for each other.
To show this, we consider well-formed configurations for which there exists a linear order, which encodes a \emph{potential creation order} of threads.
The idea is that the configuration appears as if the greatest thread is a parent thread that has itself forked \emph{all} the other threads in the configuration; the smallest thread is the child that was forked first, then the second child etc.
A child can only refer to siblings that were forked earlier, so are smaller in the order; the parent can refer to any of the children.
Note that the threads might not have been created in this (or any other) linear order,
and there may have been more complex parent-child relationships.
The operational semantics does not depend on the creation order and there may be multiple linear orders that are all consistent with a given configuration.
\begin{definition}
  Consider a configuration $C=\config w\waiter \threads$.
  Consider a linear order $<$ on $w$,  and a type $A$. Let $w'\subseteq \Tids$ be disjoint from $w$.
(The idea is that $<$ is the creation order, and $w'$ describes some threads not in the current pool, which may be useful when we are zooming in on single threads.)
  We say $C$ is \emph{well-formed} for $(A,<,w')$ if
  \begin{itemize}
  \item The waiting order $\waiter$ is transitive.
  \item A thread only waits on threads in the pool or in $w'$: if $a\prec b$ then $a\in w\uplus w'$.
  \item Threads only wait for siblings that were created earlier:
    if $a\waiter b$ and $a,b\in w$ then $a<b$.
  \item All threads that have not yet finished have type $A$, and only rely on previously created siblings:
    for all $a\in w$ we have $\cterm {\{b<a\}\uplus w'} \threads(a):A$; (the parent i.e.~greatest in $<$, can rely and wait on all its children).
  \end{itemize}\end{definition}

  \begin{proposition}[Preservation]
    Let $C_1=\config {w_1}{\waiter_1}{\threads_1}$ and
    $C_2=\config {w_2}{\waiter_2}{\threads_2}$.
    If $C_1$ is well-formed for $(A,<_1,w')$ and $C_1\xlongrightarrow{(\sigma)} C_2$
    and $w'$ is disjoint from $w_2$, then there is a linear order $<_2$ extending $<_1$
    such that $C_2$ is well-formed for $(A,<_2,w')$.
  \end{proposition}
  \begin{proof}[Proof notes]
    By induction on the derivation of transitions.
  \end{proof}
  \subsubsection{Labelled Posets Uniquely Determined from Terms of Empty Type }

\begin{proposition}\label{prop:determinacy}
  \begin{enumerate}
    Consider a term $\cterm{\emptyset}t:\tzero$ of empty type in the empty world $\emptyset$.
  \item
    The configuration $\sconfig{t}$
    reaches a terminal configuration, i.e.~there exists a labelled poset $(P,\ell)$ such that
    $
      \sconfig{t}\Downarrow (P,\ell)\text.
    $
  \item
      If
  $\quad
    \sconfig{t}\Downarrow (P_1,\ell_1)
    \quad\text{and}\quad
    \sconfig{t}\Downarrow (P_2,\ell_2)\text,\quad
  $
  then the labelled posets are isomorphic, i.e.~there is an order isomorphism
  $f:P_1\cong P_2$ such that $\ell_1(e)=\ell_2(f(e))$.
  \end{enumerate}
  \end{proposition}
  \begin{proof}[Proof notes]
    Part~1 holds in greater generality: every reduction sequence
    starting in a well-typed term $\cterm\emptyset t : A$ is
    finite. Our proof uses a straightforward combination of Tait's
    method~\cite{taits-method} and \Konig's tree
    lemma~\cite{konigs-lemma}.  Each reduction sequence of a program
    induces a finitely-branching tree in which each branch corresponds
    to the sequential execution of a single thread that does not
    mention the other threads. These thread-local executions include
    transitions steps in which the environment changes the status of a
    known tid to $\stopped$. Each infinite reduction sequence induces
    an infinite such tree, and \Konig's tree lemma implies it has an
    infinite branch. We then use Tait's method, i.e., design an
    appropriate Kripke logical relation, that shows that in all
    well-typed programs every thread has only finite sequential
    executions. The Kripke property of the relation is with respect to
    injective relabelling of tids. We define two `value' relations:
    one indexed by types over closed values, and the other indexed by
    contexts over closed environments.  The computation relation,
    indexed by types, over closed computations states that the
    computation has no infinite reduction sequence, and whenever the
    computation evaluates to $\ret v$, the value $v$ satisfies the
    appropriate value relation. We then prove the Fundamental Property
    of these relations: every well-typed value, computation, and
    substitution maps closed environments satisfying the value
    relation to values, computations, and closed environments
    satisfying the relation.

    For part 2, we prove a confluence result.
    First, we pick a deterministic naming scheme for the fresh thread IDs introduced by $\gfork()$,
    so that fresh thread IDs are independent of the evaluation order.
    One good scheme (e.g.~\cite{mn-async-ccs}) is that a thread ID is a finite sequence of numbers, with the idea that the ID
    $(m_1m_2m_3\dots m_k m_{k+1})$ is the $m_{k+1}$th thread spawned directly by the thread $(m_1\dots m_k)$.

We then show that
\begin{enumerate}
\item If $C\xlongrightarrow{(\sigma)}_a C_1$ and $C\xlongrightarrow{(\sigma)}_a C_2$ then $C_1=C_2$; and
\item If $C\xlongrightarrow{(\sigma)}_a C_1$ and $C\xlongrightarrow{(\tau)}_b C_2$ then
  there is $C'$ such that $C_1\xlongrightarrow{(\tau)}_b C'$ and $C_2\xlongrightarrow{(\sigma)}_a C'$.
\end{enumerate}
The first is local determinacy within each thread, which is straightforwardly proved by induction on the structure of transition derivations.
The second is also proved by induction on the structure of transition derivations.
However, some care is needed that the transitive closure in the local-to-global rule for a step in a particular enabled thread does not introduce dependencies that would cause a different currently enabled thread to have to wait.
Here the key strengthening of the induction hypothesis is to show that
\[\text{If }\config w \waiter \threads\xlongrightarrow {(\sigma)}_a\config {w'}{\waiter'}{\threads'}
\text{and $c\waiter' b$ and $b\in w$ then either $c\waiter b$ or $a=b$ or $a\waiter b$.}\]
\end{proof}

\section{Parameterized Algebraic Theories, Illustrated via a New Theory of Labelled Posets}\label{sec:param-algebr-theor}

Algebraic effects are formalized using algebraic theories from universal algebra.
This section recalls the concepts of algebraic theories and their generalization, parameterized algebraic theories, along with a novel running example, the theory of labelled posets.

\subsection{Algebraic Theories}

\begin{definition}
  A (first-order finitary) \emph{algebraic signature} $\MCO = \langle |\MCO|, \ar\rangle$ is a collection of operations $|\MCO|$
  and a function $\ar: |\MCO| \to \mathbb{N}$, associating a natural number to each operation, called its \emph{arity}.
\end{definition}

  We write $\mathsf{O} : n$ for an operation $\mathsf{O}$ with arity $n$.
  A \emph{context} $\Delta = a_1,\sdots,a_n$ is a list of distinct variables.
  \emph{Terms} in a context $\Delta$ are inductively generated by:
  \begin{mathpar}
    \inferrule{ }{\Delta, a, \Delta' \vdash a}
    \and
    \inferrule{\Delta \vdash u_1 \quad \cdots \quad \Delta \vdash u_n \and \mathsf{O} : n}{\Delta \vdash \mathsf{O}(u_1, \sdots, u_n)}
  \end{mathpar}

  \begin{definition}\label{def:alg-theory}
    A (first-order finitary) \emph{algebraic theory} $\MCT = (\MCO, E)$ is a pair of an algebraic signature $\MCO$ and a set of \emph{equations} $E$, where an equation is a pair of terms in the same context, $\Delta\vdash t_1$ and $\Delta\vdash t_2$,
    which we write $\Delta\vdash t_1=t_2$.
  \end{definition}


\begin{example} \label{ex:semi-lattices}
  The algebraic theory of semi-lattices $\MCL$ has two operations: $\oplus : 2$ and $\emptid : 0$. The equations are that $\oplus$ is associative, symmetric, idempotent, and has $\emptid$ as its unit:
  \begin{mathpar}
    a,b,c \vdash (a\oplus b) \oplus c = a\oplus(b\oplus c)
    \and
    a,b \vdash a\oplus b = b\oplus a
    \and
    a \vdash a\oplus a = a
    \and
    a \vdash a\oplus \emptid = a
  \end{mathpar}
  The theory of semi-lattices is often used as a semantics for non-deterministic choice with failure~\cite{moggi_notions_1991}.
  Terms modulo equations in a context $\Delta$ correspond to  subsets of the variables from $\Delta$.
\end{example}

\subsection{Parameterized Algebraic Theories}

Parameterized algebraic theories are an extension of plain algebraic theories that allow
the binding of abstract parameters. They have been
used to axiomatize effects that involve a kind of resource, such as
new memory locations in local state and channels in the $\pi$-calculus~\cite{Staton13}.

A parameterized algebraic theory is parameterized over an ordinary algebraic theory in the sense of~\Cref{def:alg-theory}.
For the rest of the paper, we fix this ordinary algebraic theory to be the theory of semi-lattices from~\Cref{ex:semi-lattices}.
We recall the definition of parameterized algebraic theories along with a running example of a novel theory of labelled posets.

\begin{definition}
  A \emph{parameterized signature} $\MCO = \langle |\MCO|, \ar\rangle$ is a collection of operations $|\MCO|$
  along with a function $\ar: |\MCO| \to \mathbb{N}\times \mathbb{N}^*$,
  associating to each operation $\mathsf{O}$ an arity consisting of a natural number and a list of natural numbers: $\ar(\mathsf{O})=(p \mid m_1, \sdots, m_k)$.
  This means $\mathsf{O}$ takes $p$ parameters and $k$ continuations, binding $m_i$ parameters in the $i$th continuation.
\end{definition}

\begin{example}\label{exa:act-sig}
  Consider operations $\node_\sigma : (1 \mid 1)$ and  $\posetend : (0 \mid )$ where $\sigma$ ranges over a fixed set of observable actions $\Sigma$.
  The $\node_\sigma$ operation takes one free parameter and one continuation binding one parameter variable; $\posetend$ takes zero parameters and no continuations.
  A parameter stands for a term in the theory of semi-lattices.
\end{example}

A parameterized context $\Gamma \mid \Delta$ has two components: $\Delta$ is a list of \emph{parameter variables} i.e.~an ordinary algebraic context in the underlying parameterizing theory; in our case this is the theory of semi-lattices and parameters are thread IDs.
The component $\Gamma$ is a list of distinct \emph{computation variables} $\Gamma = x_1 : m_1,\sdots,x_k:m_k$, where each $x_i$ is annotated with the number $m_i$ of parameters it uses; $x_i$ is a variable for which we can substitute a term of the parameterized theory.
  We often also refer to $x_i$ as a continuation.
  Terms in context $\Gamma\vbar \Delta$ are inductively generated by:
  \begin{displaymath}
    \inferrule{\left(\Delta\vdash u_i\right)^m_{i=1}}{\Gamma, x:m,\Gamma' \vbar  \Delta \vdash x(u_1,\sdots, u_m)}
    \hspace{0.4cm}
    \inferrule
    { \left(\Delta\vdash u_i\right)^p_{i=1} \and \left(\Gamma \vbar \Delta, b_1,\sdots, b_{m_i}\vdash t_i\right)^k_{i=1} \and \mathsf{O} :(p\vbar m_1,\sdots, m_k)}
    {\Gamma \vbar \Delta \vdash \mathsf{O}(u_1,\sdots, u_p,\ b_1\sdots b_{m_1}.t_1,\ \dots,\ b_1\sdots b_{m_k}.t_k)}
  \end{displaymath}
  where $\Delta\vdash u_i$ is a term judgement in the ordinary algebraic theory of semi-lattices from~\Cref{ex:semi-lattices}.
  Both contexts admit all the usual structural rules and we treat all terms up to renaming of variables.

Using the signature from~\Cref{exa:act-sig} we can build the following terms in context:
\begin{align}
  &x: 1 \mid a \vdash \node_\sigma(a, b.x(b)) \label{eq:1}
  \\
  &x:2 \mid a_1, a_2 \vdash \node_\sigma(a_1\oplus a_2, b_1.\node_\tau(a_1, b_2.x(b_2, b_1))) \label{eq:2}
\end{align}
In the term~(\ref{eq:1}), $a$ is a free parameter while $b$ is bound in $x$.
From a concurrency perspective, we interpret $\node_\sigma(a, b.x(b))$ as
forking a new child thread that performs action $\sigma$ after the thread with ID $a$ has performed its action.
The thread ID of the child performing $\sigma$ is $b$ and the continuation $x(b)$ is executed concurrently.

The algebraic operation $\node_\sigma$ is the counterpart to the generic effect $\gnode{\sigma}$ discussed in~\Cref{sec:simple-compl-fragm}.
We can encode one in terms of the other as follows:
\begin{align*}
  &\gnode{\sigma}(a)\ =\ \node_\sigma(a, b.\ret{b}) \\
  &\node_\sigma(a, b.t)\ =\ \letin{b}{\gnode{\sigma}(a)}{t}
\end{align*}
where $a$ now stands for a set of thread IDs thanks to our use of the semi-lattice theory, and thus the type of $\gnode{\sigma}$ is $\tid\to\tid$.

The term~(\ref{eq:2}) uses the operation $\oplus$ from the theory of semi-lattices to wait on both thread IDs $a_1$ and $a_2$ before executing $\sigma$.
\Cref{fig:node-ops}~(c) is a graphical representation of the term~(\ref{eq:2}), where $a_1$ and $a_2$ are inputs at bottom, and the two parameters that $x:2$ takes are outputs at the top.
The names of bound parameters $b_1$ and $b_2$ do not appear.
The solid line signifies causal dependency.
Terms built using $\node$ and $\posetend$ contain at most one computation variable, so the name of this variable is not recorded in the graphical representation.

We define two substitution operations on terms, one for parameters variables and one for computation variables, in the standard capture-avoiding way as to admit the following rules:
\begin{equation}\label{eq:subst-rules}
  \inferrule
  {\Gamma\vbar\Delta,a\vdash t \and \Delta\vdash u}
  {\Gamma\vbar\Delta\vdash t[u/a]}
  \qquad\qquad
  \inferrule{\Gamma,x:m\vbar \Delta\vdash t \and \Gamma\vbar\Delta,b_1,\sdots, b_m\vdash s}{\Gamma\vbar\Delta \vdash t[b_1\sdots b_m.s/x] }
\end{equation}
The notation $b_1\sdots b_m.s/x$ emphasises that the bound parameters $b_1,\sdots,b_m$ in $s$ will be replaced with the parameters passed to $x$.

Below are examples of each kind of substitution.
They can be understood graphically: the first transforms \Cref{fig:node-ops}~(b) into \Cref{fig:node-ops}~(c) and the second transforms \Cref{fig:node-ops}~(c) into \Cref{fig:node-ops}~(d).
{\footnotesize
\begin{align}
  \node_\sigma(a_3, b_1.\node_\tau(a_1, b_2.x(b_2, b_1)))[a_1\oplus a_2 / a_3] &= \node_\sigma(a_1\oplus a_2, b_1.\node_\tau(a_1, b_2.x(b_2, b_1))) \label{eq:3}\\
  \node_\sigma(a_1\oplus a_2, c_1.\node_\tau(a_1, c_2.x(c_2, c_1)))[b_1b_2.y(b_1\oplus b_2) / x] &= \node_\sigma(a_1\oplus a_2, c_1.\node_\tau(a_1, c_2.y(c_2 \oplus c_1))) \label{eq:4}
\end{align}
}

\begin{definition}
  A \emph{parameterized algebraic theory} $\MCT = (\MCO, E)$ is a pair of a parameterized signature $\MCO$ and a set $E$ of equations.
  An equation is a pair of terms in the same context $\Gamma\mid \Delta$, which we write as $\Gamma\mid\Delta\vdash t_1 = t_2$.
\end{definition}

\begin{example}\label{exa:act-theory}
  The parameterized \emph{theory of labelled posets}, $\actth$, consists of the signature from~\Cref{exa:act-sig}, containing $\node_\sigma$ and $\posetend$, and of the following two equations:
  \begin{align}
    x: 2 \mid a_1, a_2 &\vdash \node_\sigma(a_1, b_1.\node_\tau(a_2, b_2.x(b_1, b_2))) = \node_\tau(a_2, b_2.\node_\sigma(a_1, b_1.x(b_1, b_2))) \label{eq:node-comm} \\
    x: 1 \mid a &\vdash \node_\sigma(a, b.x(b)) = \node_\sigma(a, b.x(a\oplus b)) \label{eq:node-trans}
  \end{align}
  The first equation states that independent actions may happen in any order.
  The second equation encodes the transitivity of causal dependencies.
  There are no equations involving $\posetend$; intuitively $\posetend$ finishes the execution of the whole program.
\end{example}

In~\Cref{sec:pres-theory-fork} we will present an extended example of a parameterized algebraic theory for forking threads, together with examples of equational reasoning in~\Cref{exa:fork-wait-perform,exa:fork-wait-act,exa:fork-wait-main-thread}.

Given a parameterized theory $\MCT = (\MCO, E)$, we can form an equivalence relation $=_\MCT$ on the terms of $\MCT$, called \emph{derivable equality}, by closing all \emph{simultaneous substitution instances} of the equational axioms from $E$ under reflexivity, symmetry, transitivity, and two congruence rules, one for variables and one for the operations in $\MCO$.
Below is the congruence rule for variables; it allows us to use the equations of the theory of semi-lattices when reasoning about parameterized terms.
\begin{displaymath} 
  \inferrule{ \left(\Delta\vdash u_i=u'_i\right)^m_{i = 1}}{\Gamma, x:m \vbar  \Delta \vdash x(u_1, \sdots, u_m)=_\MCT x(u'_1, \sdots, u'_m)}
\end{displaymath}


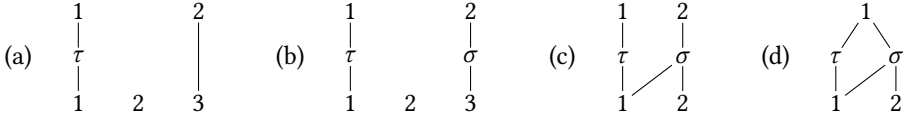
\begin{figure}
  \begin{minipage}{0.95\linewidth}
  \begin{center}
    \begin{tikzpicture}[yscale = .75, inner sep=1pt]
      \node (x1)  at (0.0, 1.6) {$1$};
      \node (x2)  at (1.6, 1.6) {$2$};
      \node (tau) at (0.0, 0.8) {$\tau$};
      \node (a1)  at (0, 0)     {$1$};
      \node (a2)  at (0.8, 0)   {$2$};
      \node (a3)  at (1.6, 0)   {$3$};

      \node (a) at (-0.8, 0.8) {(a)};

      \path[-]
        (a1)  edge (tau)
        (tau) edge (x1)
        (a3)  edge (x2);
    \end{tikzpicture}
    \hspace{0.7cm}
    \begin{tikzpicture}[yscale = .75, inner sep=1pt]
      \node (x1)  at (0.0, 1.6) {$1$};
      \node (x2)  at (1.6, 1.6) {$2$};
      \node (tau) at (0.0, 0.8) {$\tau$};
      \node (sig) at (1.6, 0.8) {$\sigma$};
      \node (a1)  at (0, 0)     {$1$};
      \node (a2)  at (0.8, 0)   {$2$};
      \node (a3)  at (1.6, 0)   {$3$};

      \node (b) at (-0.8, 0.8) {(b)};

      \path[-]
        (a1)  edge (tau)
        (tau) edge (x1)
        (a3)  edge (sig)
        (sig) edge (x2);
    \end{tikzpicture}
    \hspace{0.7cm}
    \begin{tikzpicture}[yscale = .75, inner sep=1pt]
      \node (x1)  at (0, 1.6)   {$1$};
      \node (x2)  at (0.8, 1.6) {$2$};
      \node (tau) at (0, 0.8)   {$\tau$};
      \node (sig) at (0.8, 0.8) {$\sigma$};
      \node (a1)  at (0, 0)     {$1$};
      \node (a2)  at (0.8, 0)   {$2$};

      \node (c) at (-0.8, 0.8) {(c)};

      \path[-]
        (a1)  edge (tau)
        (a1)  edge (sig)
        (tau) edge (x1)
        (a2)  edge (sig)
        (sig) edge (x2);
    \end{tikzpicture}
    \hspace{0.7cm}
    \begin{tikzpicture}[yscale = .75, inner sep=1pt]
      \node (y1)  at (0.4, 1.6)   {$1$};
      \node (tau) at (0, 0.8)   {$\tau$};
      \node (sig) at (0.8, 0.8) {$\sigma$};
      \node (a1)  at (0, 0)     {$1$};
      \node (a2)  at (0.8, 0)   {$2$};

      \node (d) at (-0.8, 0.8) {(d)};

      \path[-]
        (a1)  edge (tau)
        (a1)  edge (sig)
        (tau) edge (y1)
        (a2)  edge (sig)
        (sig) edge (y1);
    \end{tikzpicture}
  \end{center}
  \end{minipage}
  \caption{
    Graphical examples of terms built from the $\node$ operation.
    (a)~is the term $\node_\tau(a_1,b.x(b,a_3))$; numbers $1$ and $2$ at the top correspond to the two inputs of variable $x:2$.
    (b)~is the application of $\node_\sigma(a'_3, a_3.{-})$ to (a).
    (c)~is the term in~\cref{eq:2}; it is obtained from (b) by substituting  $a_1 \oplus a_2$ for $a'_3$ as in~\cref{eq:3};
    (d)~is obtained from (c) by substituting a term for the computation variable $x:2$, as in~\cref{eq:4}.
  }
  \label{fig:node-ops}
\end{figure}

\subsection{Models of Parameterized Algebraic Theories}

We recall models of parameterized algebraic theories by analogy with models for ordinary algebraic theories.
For this paper, it is sufficient to consider the case where the parameterizing theory is that of semi-lattices~(\Cref{ex:semi-lattices}), but more general notions of models exist~\cite{Staton13, Staton13PL, Staton15}.
In~\Cref{sec:model-caus-struct}, we illustrate models by considering the theory of labelled posets from~\Cref{exa:act-theory}.

\subsubsection{Connection Between the Category of Finite Sets and Relations and the Theory of Semi-lattices}\label{sec:semilat-terms-interp}

We define the objects of the category $\indexcat$ to be natural numbers $n$ and morphisms $n\rightarrow n'$ to be relations $R \subseteq [n] \times [n']$, where $[n]$ denotes the set $\{1,\sdots,n\}$.
Composition $S\circ R$ of $R \subseteq [n] \times [n']$ with $S \subseteq [n'] \times [n'']$ is the usual composition of relations.
We can also think of $\FinRel$ as the Kleisli category of the powerset monad, restricted to finite sets.

For each $p$, we can define an isomorphism between the set of (equivalence classes of) terms $\{[a_1,\sdots,a_p\vdash u]\}$ in the theory of semi-lattices and the set of morphisms $\indexcat(1,p)$.
(In fact, $\FinRel$ is the opposite category to the Lawvere theory of semi-lattices.)
\begin{mathpar}
  \den{a_1,\sdots,a_p\vdash a_i} = \{(1, i)\} \and
  \den{a_1,\sdots,a_p\vdash \emptid} = \emptyset \and
  \den{a_1,\sdots,a_p\vdash u_1\oplus u_2} = \den{a_1,\sdots,a_p\vdash u_1} \cup \den{a_1,\sdots,a_p\vdash u_2}
\end{mathpar}

\subsubsection{Models of Parameterized Theories in $\Set^{\indexcat}$}\label{sec:models-param-theor}

A model of an ordinary algebraic theory consists of a set, the carrier, together with structure for interpreting the operations in the theory, such that all equational axioms are satisfied.
We are studying theories parameterized by the algebraic theory of semi-lattices, so we will consider models where, instead of a set, the carrier is a family of sets indexed by the objects of the category $\indexcat$.

\begin{definition}\label{def:sigma-struct}
  Let $\algsig$ be a parameterized signature.
  An \emph{$\algsig$-structure} $\MCX$ is an object $X$ in the functor category $\Set^{\indexcat}$ equipped with for each operation $\mathsf{O} :(p\vbar m_1,\sdots, m_k)$ a family of functions indexed by natural numbers $n$, and respecting naturality with respect to morphisms in $\indexcat$:
  \begin{displaymath}
    \mathsf{O}_{\MCX,n} : X(n+m_1)\times\ldots \times X(n+m_k)\rightarrow X(n+p)
  \end{displaymath}
\end{definition}

Given an $\algsig$-structure $\MCX$, the interpretation of operations can be extended to all terms using the interpretation of semi-lattices terms from~\Cref{sec:semilat-terms-interp}.
A term $x_1:m_1,\sdots,x_k:m_k \vbar a_1,\sdots,a_p\vdash t$ is interpreted as a family of functions
  \begin{displaymath}
    \den{t}_{\MCX,n} : X(n+m_1)\times \ldots \times X(n+m_k)\rightarrow X(n+p)
  \end{displaymath}
  natural in $n$, defined by structural recursion.
  A computation variable
  \begin{displaymath}
    x_1:m_1,\sdots,x_k:m_k \vbar a_1,\sdots,a_p\vdash x_i(u_1,\sdots,u_{m_i})
  \end{displaymath}
  is interpreted as projection followed by the interpretation of its semi-lattice term inputs
  \begin{displaymath}
    X(n+m_1)\times\ldots\times X(n+m_k) \xrightarrow{\pi_i} X(n+m_i) \xrightarrow{X(n+[\den{u_1},\dots,\den{u_{m_i}}])} X(n+p)
  \end{displaymath}
  where $\den{u_1},\sdots,\den{u_{m_i}} :1\rightarrow p$ are morphisms in $\indexcat$.

  A term of the form $\Gamma\vbar a_1,\sdots,a_p \vdash \mathsf{O}(u_1,\sdots, u_{p'},\ b_1\sdots b_{m'_1}.t_1,\ \dots,\ b_1\sdots b_{m'_{k'}}.t_{k'})$ is interpreted as the $n$-indexed family,
  \begin{displaymath}
    X(n+[\mathsf{id}_p,\den{u_1},\sdots,\den{u_{p'}}]) \circ \mathsf{O}_{\MCX,n+p} \circ \langle \den{t_1}_{\MCX,n},\sdots, \den{t_{k'}}_{\MCX,n} \rangle\text,
  \end{displaymath}
  where $\mathsf{O} :(p'\vbar m'_1,\sdots, m'_{k'})$.
  The map $[\mathsf{id}_p,\den{u_1},\sdots,\den{u_{p'}}] : p+p'\rightarrow p$ interprets the $p'$ arguments of $\mathsf{O}$ using the parameter variables $a_1,\sdots,a_p$.

\begin{definition}\label{def:param-model}
  Let $\MCT=(\algsig,E)$ be a parameterized theory.
  A $\algsig$-structure $\MCX$ is a \emph{model} for the theory $\MCT$ if for every equational axiom from $E$, $\Gamma\vbar\Delta\vdash s=t$, and for every natural number $n$, the following functions are equal:
  \begin{displaymath}
    \den{\Gamma\vbar\Delta\vdash s}_{\MCX,n} = \den{\Gamma\vbar\Delta\vdash t}_{\MCX,n}.
  \end{displaymath}
\end{definition}

\begin{proposition}
  For a parameterized theory $\MCT$, the derivable equality $=_\MCT$ is sound: if $s=_\MCT t$ is derivable, then $\den{s}_\MCX = \den{t}_\MCX$ in any $\MCT$-model $\MCX$.
\end{proposition}

\subsubsection{A Model of the Theory of Labelled Posets}\label{sec:model-caus-struct}

To illustrate the notion of model from the previous section, we build a model for the parameterized algebraic theory of labelled posets from~\Cref{exa:act-theory}.
To do this we generalize the notion of $\symb$-labelled poset from~\Cref{def:labelled-poset}.

\begin{definition}\label{def:io-lab-poset}
  An \emph{$n$-input $m$-output $\symb$-labelled poset} $P = \langle V_P, \leq_P, l_P\rangle$ consists of a set $V_P$ of elements labelled by a function $l_P:V_P\rightarrow\symb$, and a partial order $\leq_P$ on the set $[n]\uplus V_P \uplus [m]$, such that the $n$ input elements are minimal and the $m$ output elements are maximal.
\end{definition}

  Examples of such posets appear in~\Cref{fig:node-ops}, with inputs at the bottom and outputs at the top.
  If there are no inputs and outputs,~\Cref{def:io-lab-poset} reduces to that of an ordinary $\symb$-labelled poset.
  An isomorphism between two $n$-input $m$-output $\symb$-labelled posets $P$ and $Q$ is a bijection $f: V_P \to V_Q$ that preserves the labels,
  and such that $\id_{[n]} \uplus f \uplus \id_{[m]}$ preserves and reflects the order.

For each natural number $n$, define the set $\noderepm m(n)$ to contain \emph{isomorphism classes of $n$-input $m$-output $\Sigma$-labelled posets}.
Given a relation $R\subseteq [n]\times [n']$, $\noderepm{m}(R)$ acts on a labelled poset by updating $i \leq e$ to $i' \leq e$ if $(i,i') \in R$.
The poset in~\Cref{fig:node-ops}~(c) is obtained from~\Cref{fig:node-ops}~(b) via this action.

For each natural number $m$, we equip $\noderepm m$ with a family of operations $\node_{\sigma,m}$, one for each $n$:
\begin{displaymath}
  (\node_{\sigma,m})_n : \noderepm m(n+1) \to \noderepm m (n+1)
\end{displaymath}
which given a labelled poset, labels its $(n+1)$-th input by $\sigma$ and creates a new $(n+1)$ input just below $\sigma$.
The poset in~\Cref{fig:node-ops}~(b) is the result of applying $(\node_{\sigma,2})_3$ to~\Cref{fig:node-ops}~(a).

\paragraph{Remark}
Labelled posets can be organized into a PROP (see~\cite{MARKL200887}),
where morphisms $n\to m$ are labelled posets $\noderepm{m}(n)$, identities are given by the poset with no labelled elements,
composition ``plugs'' the outputs of a poset into the inputs of another,
and monoidal composition is juxtaposition.
This categorical formulation was valuable for proving~\Cref{prop:lab-poset-model} and~\Cref{thm:lab-poset-free-model}.
The same PROP appears in~\cite{Fiore2013} as the PROP of L-labelled transitive idags, where it is characterized as being free on a suitable equational theory.
In~\cite{DBLP:journals/corr/Mimram15}, this PROP is shown to be the monoidal category presented by a certain rewriting system.
Similar categorical ideas appear elsewhere, e.g.~\cite{bsz-sig-flow,baez-prop-network,iposets},
but typically with a first-order emphasis, whereas we are aiming at a semantics for
programming language via monads~(\S\ref{sec:denot-semant-progr}).

  For each context $\Gamma$ of computation variables, we construct a functor $\noderep{\Gamma}:\FinRel\to\Set$ that will be the carrier of a model of the theory of labelled posets, in the sense of~\Cref{def:sigma-struct}.
  For each natural number $n$, define the set
  \begin{displaymath}
    \noderep{\Gamma}(n) = \biguplus_{x:m\in\Gamma} \noderepm m(n) \uplus \noderepm{0}(n).
  \end{displaymath}
  We equip $\noderep\Gamma(n)$ with an operation $(\node_{\sigma})_n : \noderep\Gamma(n+1) \to \noderep\Gamma(n+1)$ by applying $(\node_{\sigma,m})_n$ pointwise.
  Let $\posetend_{n}: 1 \to \noderep\Gamma(n)$ be the function that selects, from the right injection $\noderepm{0}(n)$, the labelled poset with only the $n$ inputs as elements and with discrete order.

\begin{proposition}\label{prop:lab-poset-model}
  For each context $\Gamma$, the functor $\noderep{\Gamma}$, together with the natural transformations $\node_\sigma$ and $\posetend$ defined above,
  is a model of the parameterized theory of labelled posets from~\Cref{exa:act-theory}.
\end{proposition}

\subsection{Free Models of Parameterized Algebraic Theories}\label{sec:free-models}

We now return to the study of models of parameterized algebraic theories in general.
Using the evident notion of homomorphism between $\algsig$-structures, we can discuss $\MCO$-structures and $\MCT$-models that are \emph{free} over a collection $X\in\Set^\indexcat$ of generators.

\begin{definition}\label{def:free-model}
  Consider a $\MCT$-model $\MCY$ with carrier $Y\in\Set^{\indexcat}$ and a morphism $\eta_X:X\rightarrow Y$ in $\Set^{\indexcat}$.
  The model $\MCY$ is \emph{free on $X$} if for any other model $\MCZ$ and any morphism $f:X\rightarrow Z$ in $\Set^\indexcat$,
  there exists a unique homomorphism of models $\hat{f}:\MCY\rightarrow\MCZ$ that extends $f$, meaning $\hat{f}\circ\eta_X=f$ in $\Set^{\indexcat}$.
\end{definition}

Given a context of computation variables $\Gamma$, consider the functor $V_\Gamma$ where:
\begin{displaymath}
  V_\Gamma(n)= \left\{[x(u_1,\sdots,u_{m})]_{=_\MCT}\ \big|\ \Gamma\vbar a_1,\sdots,a_n \vdash x(u_1,\sdots,u_{m}) \right\}.
\end{displaymath}
The equivalence relation on terms in $V_\Gamma$ is non-trivial because the parameter terms $u_i$ are quotiented by the semi-lattice equations.

The \emph{term model} is given by the functor $F_\MCT(V_\Gamma)$, which contains equivalence classes of terms:
\begin{equation}\label{eq:equiv-clas-terms}
  F_\MCT(V_\Gamma)(n)=\left\{[t]_{=_\MCT} \ \big| \ \Gamma\vbar a_1,\sdots,a_n\vdash t \right\}
\end{equation}
The action on morphisms $n\rightarrow n'$, which are relations $R \subseteq [n]\times [n']$, is given by substitution of parameters.
The functor $F_\MCT(V_\Gamma)$ can be made into a $\MCT$-model using the syntactic term formation rules,
and we can construct a morphism $\eta_{V_\Gamma}:V_\Gamma\rightarrow F_\MCT(V_\Gamma)$ by embedding variables into terms.
We use the term model to prove the completeness result below.

\begin{proposition}
  ~
  \begin{enumerate}
  \item The functor $F_\MCT(V_\Gamma)$ is a $\MCT$-model and is moreover a free $\MCT$-model on $V_\Gamma$.
  \item The derivable equality $=_\MCT$ in a parameterized algebraic theory is complete: if an equation is valid in every $\MCT$-model then it is derivable in $=_\MCT$.
  \end{enumerate}
\end{proposition}

We can now characterize the labelled posets model from~\Cref{prop:lab-poset-model} using the universal property of a free model.
In particular, equivalence classes of closed terms $\{[\cdot\mid \cdot \vdash t]_{=_\MCC}\}$ built from $\node$ and $\posetend$~(\Cref{exa:act-theory}) are in bijection with ordinary $\symb$-labelled posets.

\begin{theorem}\label{thm:lab-poset-free-model}
  For each context $\Gamma$, the functor $\noderep{\Gamma}$ together with the natural transformations $\node_\sigma$ and $\posetend$ is isomorphic to the free model $F_\actth(V_\Gamma)$ of the theory of labelled posets from~\Cref{exa:act-theory}.
\end{theorem}

\begin{proof}[Proof notes]
  An interpretation homomorphism $F_\actth(V_\Gamma) \to \noderep{\Gamma}$ is given by the unique map from the free model.
  We define an inverse map $\noderep{\Gamma} \to F_\actth(V_\Gamma)$ that linearizes a labelled poset into a nested $\node_\sigma(u, a.{-})$ term, with one $\node_\sigma$ operation for each element of the poset labelled by action $\sigma$.
  Fixing a poset element, all the elements preceding it in the poset order are encoded by the compound thread ID $u$.
  Equation~\eqref{eq:node-comm} ensures the map $\noderep{\Gamma} \to F_\actth(V_\Gamma)$ is independent of the choice of linearization, and equation~\eqref{eq:node-trans} ensures that all terms are $(=_\actth)$-equal to a term in the image of this map, by asking that causal dependencies are transitively closed.
  The proof strategy is similar to \Cref{thm:main-theorem}, for which we provide more detail.
\end{proof}

\section{A Parameterized Algebraic Theory of Dynamic Threads}\label{sec:pres-theory-fork}

In this section we introduce an equational axiomatization for the concurrency features from~\Cref{sec:background} ($\gfork$, $\gwait$, $\gsop$ and $\gprintstop{\sigma}$) as a parameterized algebraic theory.
In~\Cref{sec:rep-theorem} we interpret this theory semantically, using labelled posets similar to those from~\Cref{sec:model-caus-struct}.
Then in~\Cref{sec:denot-semant-progr} we extend the semantics to model the whole concurrent programming language from~\Cref{sec:background}.

\subsection{Presentation of the Theory}\label{sec:presentation-theory}

\subsubsection{Signature}
We introduce a \emph{theory of dynamic threads} $\forkth$, parameterized by the theory of semi-lattices, with the following signature, where $\sigma$ ranges over a fixed set of observable actions $\Sigma$:
  \begin{mathpar}
    \fork :(0 \vbar 1,0) \and \wait:(1\vbar 0) \and \sop:(0\vbar) \and \act{\sigma}:(0 \vbar)
  \end{mathpar}

In the term $\fork(a.x(a),y)$ the variable $x$ is the parent thread, while $y$ is the child thread; the parameter $a$ is the thread ID of the child $y$ and is bound in $x$.
The parent might wait for the child named $a$ to finish, then continue as $z$, using the operation $\wait(a,z)$.
The operation $\sop$ has no continuation and indicates that the current thread has finished execution; $\act{\sigma}$ performs the observable action $\sigma$ and finishes.
Parameters carry a semi-lattice structure, so it is possible to wait on a compound thread ID, e.g.~$\wait(a_1\oplus a_2,z)$ waits for both $a_1$ and $a_2$, or on no thread ID at all, $\emptid$.

\begin{example}\label{exa:fork-wait-signature}
   The term $t_1$ encodes sequential execution of action $\sigma_1$ followed by $\sigma_2$, while $t_2$ and $t_3$ encode concurrent execution of $\sigma_1$ and $\sigma_2$:
  \begin{mathpar}
    t_1 = \fork(a.\wait(a,\act{\sigma_2}),\act{\sigma_1})
    \and
    t_2 = \fork(a.\act{\sigma_2},\act{\sigma_1})
    \and
    t_3 = \fork(a.\act{\sigma_1},\act{\sigma_2})
  \end{mathpar}
  However, $t_2$ and $t_3$ have slightly different intended semantics.
  In the term $\fork(b.\wait(b,\act{\tau}),t_2)$ the ID $b$ refers only to the thread $\act{\sigma_2}$ and not to its child $\act{\sigma_1}$, so a possible execution is $\sigma_2\tau\sigma_1$.
  But this is not possible in $\fork(b.\wait(b,\act{\tau}),t_3)$ because here $\sigma_1$ must happen before $\tau$.

  More generally, in the expression $\fork(a.x(a),y)$ we often refer to $x$ as the \emph{main thread} because its ID is available to the environment to wait on, while the ID of $y$ is only available to $x$.

\end{example}

\subsubsection{Equations}
The equational axioms for the theory of dynamic threads appear in~\Cref{fig:fork-wait-eqs}.
There are no equations involving observable actions $\act{\sigma}$.
\Cref{eq:wait-closed} states that if $x$ is waiting for $a$ to finish, then waiting for $a$ in the future is redundant.
Commutativity of fork,~\cref{eq:fork-comm}, holds only if the children $y$ and $z$ do not use each other's IDs.
Similarly, associativity,~\cref{eq:fork-assoc}, holds if the parent $x$ does not use the ID $b$ of $z$.
\Cref{eq:child-becomes-main} says that forking a child $x$ and waiting for it to finish is the same as running $x$ as the main thread.
The wait is necessary because it forces the environment to wait on $x$ even when it is executed as a child thread.
\Cref{eq:fork-sub} removes a child that does not perform any observable action; it involves a substitution of parameter $b$ for $a$ in $x$.

\begin{figure}
  \framebox{
  \begin{minipage}{0.95\linewidth}
  Equations describing the interaction of $\wait$ with the semi-lattice structure of thread IDs.
  \begin{align}
    x:0 \vbar - &\vdash \wait(\emptid, x) = x \label{eq:wait-neutral-elem}\\
    x:0\vbar a,b &\vdash \wait(a, \wait(b, x)) = \wait(a \oplus b, x) \label{eq:wait-accumulates}\\
    x:1\vbar a,b &\vdash \wait(a, x(b)) = \wait(a, x(a\oplus b)) \label{eq:wait-closed}
  \end{align}
  The $\wait$ and $\fork$ operations commute; $\fork$ is commutative and associative.
  \begin{align}
    x:1,y:0\vbar b &\vdash \wait(b, \fork(a.x(a), y)) = \fork(a.\wait(b, x(a)), \wait(b, y)) \label{eq:fork-wait-comm}\\
    x:2,y:0,z:0\vbar -&\vdash\fork(a.\fork(b.x(a,b), y), z) = \fork(b.\fork(a.x(a,b), z), y) \label{eq:fork-comm} \\
    x:1,y:1,z:0 \vbar - &\vdash\fork(a.x(a), \fork(b.y(b), z)) = \fork(b.\fork(a.x(a), y(b)), z) \label{eq:fork-assoc}
  \end{align}
  The $\sop$ operation acts as a unit for $\fork$.
  \begin{align}
    x:0 \vbar - &\vdash\fork(a.\wait(a, \sop), x) = x \label{eq:child-becomes-main} \\
    x:1\vbar b &\vdash \fork(a.x(a), \wait(b, \sop)) = x(b) \label{eq:fork-sub}
  \end{align}
\end{minipage}
}
  \caption{Equations for the parameterized algebraic theory of dynamic threads $\forkth$.}
  \label{fig:fork-wait-eqs}
\end{figure}

\begin{example}\label{exa:fork-wait-perform}
  We can use $\act{\sigma}$ to encode an operation $\perform{\sigma}(x)$ which executes action $\sigma$ then continues as $x$.
    We can recover $\act{\sigma}$ from $\perform{\sigma}(x)$ by setting $x$ to be $\sop$ and using~\cref{eq:child-becomes-main}.
  \begin{displaymath}
    \perform{\sigma}(x) \defeq \fork(a.\wait(a,x),\act{\sigma})
  \end{displaymath}
  The algebraic operations $\act{\sigma}$ and $\perform{\sigma}$ are the counterpart to the generic effects $\gprintstop{\sigma}$ and $\gprint{\sigma}$ discussed in~\Cref{sec:alt-lang}.
\end{example}

\begin{example}\label{exa:fork-wait-act}
  The $\node_\sigma$ operation from~\Cref{exa:act-theory} can be encoded as:
  \begin{displaymath}
    \node_\sigma(a,b.x(b)) \defeq \fork(b.x(b),\wait(a,\act{\sigma}))
  \end{displaymath}
  \Cref{eq:node-comm} for $\node_\sigma$ amounts to commutativity of $\fork$~(\cref{eq:fork-comm}), while~\cref{eq:node-trans} can be derived:
  \begin{align*}
    \node_\sigma(a,b.x(b)) &= \fork(b.x(b),\wait(a,\fork(c.\wait(c,\sop), \act{\sigma})))  \tag{\cref{eq:child-becomes-main}} \\
                                        &= \fork(b.x(b),\fork(c.\wait(a\oplus c,\sop), \wait(a,\act{\sigma}))) \tag{\cref{eq:fork-wait-comm} and~\cref{eq:wait-accumulates}} \\
                                          &= \fork(c.x(a\oplus c),\wait(a,\act{\sigma})) \tag{\cref{eq:fork-assoc,eq:fork-sub}}
  \end{align*}
\end{example}

\begin{example}\label{exa:fork-wait-main-thread}
  In the term $t=\fork(a.\wait(a,\act{\sigma_1}),\fork(b.\sop,\act{\sigma_2}))$ thread ID $a$ only refers to the thread $\sop$ and not to its child $\act{\sigma_2}$:
  \begin{align*}
    t &= \fork(b.\fork(a.\wait(a,\act{\sigma_1}),\sop),\act{\sigma_1}) & \tag{\cref{eq:fork-assoc}}\\
                                                                  &= \fork(b.\fork(a.\wait(a,\act{\sigma_1}),\wait(\emptid,\sop)),\act{\sigma_1}) &\tag{\cref{eq:wait-neutral-elem}}\\
    &=\fork(b.\act{\sigma_1},\act{\sigma_2}) & \tag{\cref{eq:fork-sub} and~\cref{eq:wait-neutral-elem}}
  \end{align*}
\end{example}

\subsection{Normal Forms}\label{sec:normal-forms}

In this section, we identify a convenient subclass of terms, which we refer to as normal forms.
We show that all the terms in the theory of dynamic threads $\forkth$ are equal to a normal form, up to the derivable equality $=_\forkth$.
We define a normal form to be a term in context of the form:
\begin{equation}\label{eq:normal-form}
      \Gamma \vbar a_1,\sdots,a_n \vdash \fork(b_1.\sdots\fork(b_p.\wait(u_{p+1},\sop),\wait(u_p,t_p)),\sdots\wait(u_1,t_1))
\end{equation}
where each subterm $t_i$ is either an observable action $\act{\sigma}$ or a variable $x(u_{i1},\sdots,u_{im})$, for some $x:m$ in $\Gamma$.
We also require that the parameters (i.e.~compound thread IDs) $u_1,\sdots,u_{p+1}$, and the parameters occurring in each $t_i=x(u_{i1},\sdots,u_{im})$ satisfy closure conditions explained below.

Informally, a normal form consists of a parent that forks $p$ child threads, waits on some collection of thread IDs, $u_{p+1}$, then finishes.
A child must perform exactly one action, or be a free computation variable.
Thanks to the term formation rules, $u_i$ can only use thread IDs $b_1,\sdots,b_{i-1}$ that have been forked earlier, or thread IDs from the context $a_1,\sdots,a_n$.

\begin{example}\label{exa:normal-forms-ex1}
  The term $\fork(b_1.\wait(b_1,\act{\sigma_2}),\act{\sigma_1})$, from~\Cref{exa:fork-wait-signature}, is not in normal form but it is $(=_\forkth)$-equal to the following normal form:
  \begin{displaymath}
    \mathsf{nf}_1=\fork(b_1.\fork(b_2.\wait(b_1\oplus b_2,\sop),\wait(b_1,\act{\sigma_2})),\act{\sigma_1})
  \end{displaymath}
  To show this, use~\cref{eq:child-becomes-main} followed by~(\ref{eq:fork-wait-comm}) and ~(\ref{eq:wait-accumulates}) to show the subterm $\wait(b_1,\act{\sigma_2})$ equals
  \begin{displaymath}
     \wait(b_1,\fork(b_2.\wait(b_2,\sop),\act{\sigma_2})) = \fork(b_2.\wait(b_1\oplus b_2,\sop),\wait(b_1,\act{\sigma_2})).
  \end{displaymath}
\end{example}

\subsubsection{Closure Conditions for Normal Forms}\label{sec:clos-cond-norm}
The closure conditions that a term of shape~(\ref{eq:normal-form}) needs to satisfy to be a normal form are that: if ID $b_j$ appears in $u_i$, then all the IDs in $u_j$ are included in $u_i$; the analogous condition when $b_j$ appears in $u_{ik}$, where $t_i=x(u_{i1},\sdots,u_{im})$; and the IDs in $u_i$ must appear in each $u_{ik}$.
Imposing these closure conditions means that a normal form contains redundant information about dependencies between different threads, but this will help us formulate a correspondence between normal forms and labelled posets, in~\Cref{sec:lab-poset-as-nf} and~\Cref{thm:main-theorem}.

\begin{example}
  The normal form from~\Cref{exa:normal-forms-ex1} satisfies these closure conditions, as does the following term, with free IDs $a_1$ and $a_2$:
  \begin{equation}\label{eq:nf-ex-2}
    \mathsf{nf}_2=\fork(b_1.\fork(b_2.\wait(b_1\oplus b_2\oplus a_1,\sop),\wait(b_1\oplus a_1,x(b_1\oplus a_1\oplus a_2))),\wait(a_1,\act{\sigma_1}))
  \end{equation}
\end{example}

\begin{definition}\label{def:normal-form}
  Fix a context of computation variables $\Gamma$.
  For each natural number $n$, define the set $\nfset{\Gamma}(n)$ to contain normal forms, i.e.~terms of shape~(\ref{eq:normal-form}) respecting the closure conditions above, quotiented by: the equivalence relation generated by reordering of child threads that do not depend on each other, and by the semi-lattice equations on compound thread IDs.
  Moreover, $\nfset{\Gamma}$ has a functorial action on relations $R\in [n]\times[n']$ given by parameter substitution.
\end{definition}

This definition implies that two representatives of the same equivalence class of normal forms are also $(=_\forkth)$-equal in the theory of dynamic threads~(\Cref{fig:fork-wait-eqs}) because the reordering of independent child threads corresponds to~\cref{eq:fork-comm}, and $=_\forkth$ is closed under the semi-lattice equations by definition.

\begin{theorem}\label{thm:normal-form-surj}
  Every term  $\Gamma\vbar a_1,\sdots,a_n \vdash t$ in the theory of dynamic threads
  $\forkth$ is derivably equal to an equivalence class of normal forms from $\nfset{\Gamma}(n)$; this class is not a priori unique.
\end{theorem}

We prove the theorem by induction on terms using the equations from~\Cref{fig:fork-wait-eqs}.
As a corollary, normal forms map surjectively onto (equivalence classes of) terms via the identity.
We have not shown for now that every term is equal to a \emph{unique} equivalence class of normal forms, we prove this in~\Cref{cor:norm-form-unique}.

\section{A Representation Theorem for the Theory of Dynamic Threads} \label{sec:rep-theorem}

In this section we interpret the parameterized theory of dynamic threads from~\Cref{sec:pres-theory-fork} semantically, using a notion of labelled poset similar to that used in~\Cref{sec:model-caus-struct}.
In~\Cref{sec:labell-posets-exampl} we discuss labelled posets informally, then in~\Cref{sec:labell-posets-form,sec:labelled-poset-model} we give their formal definition and show that they form a free model.
In~\Cref{sec:compl-theor-theory}, we show that this model is in a certain sense complete.

\subsection{Labelled Posets with Holes by Example}\label{sec:labell-posets-exampl}

We introduce labelled posets by example and use terms from the theory of dynamic threads~(\Cref{sec:pres-theory-fork}) to motivate them.
We define labelled posets formally in~\Cref{def:repr-lab-poset,def:lab-poset-well-formed}.

\subsubsection{Labelled Posets Represent Terms}
Consider~\Cref{fig:lab-poset-ex-1}~(a): two elements of the poset are labelled by observable actions $\sigma_1$ and $\sigma_2$.
The solid lines represent causal dependencies and induce a partial order: $\sigma_1$ must happen before $\sigma_2$.
All posets contain a distinguished maximal element $\possop$ which represents the \emph{end of the main thread}; we will use $\possop$ when defining an operation analogous to $\fork$ for posets.
In~\Cref{fig:lab-poset-ex-1}~(b), $\sigma_2$ is part of the main thread but $\sigma_1$ is not, as  discussed in~\Cref{exa:fork-wait-signature}.
Elements of the poset may be labelled by computation variables, for example, $x:0$ in~\Cref{fig:lab-poset-ex-1}~(c).

A term's free thread IDs appear as minimal elements in its poset, e.g.~in~\Cref{fig:lab-poset-ex-1}~(d) they are numbered $1$ and $2$.
Bound thread IDs do not appear in the poset.
\Cref{fig:lab-poset-ex-1}~(d) depicts the poset of:
\begin{displaymath}
  x:1 \vbar a_1,a_2 \vdash \fork(b_1.\wait(b_1,x(b_1\oplus a_2)),\wait(a_1,\act{\sigma_1}))
\end{displaymath}
In the poset, there is one element labelled $x$ which has one input.
The dotted line is not part of the partial order; it represents the thread IDs passed to variable $x$ and means that $x$ \emph{may} wait for $a_2$, depending on what computation $x$ is.
The dotted line induces a relation called \emph{visibility} which is assumed to be downward-closed with respect to the partial order (the solid line) and to contain all elements below $x$ in the partial order; therefore we omit to draw dotted lines from $\sigma_1$ and $1$ into $x$.

\subsubsection{Labelled Posets are Normal Forms}\label{sec:lab-poset-as-nf}
Recall that a normal form has the shape below, where each $t_i$ is either one observable action or a computation variable from $\Gamma$:
\begin{displaymath}
      \Gamma \vbar a_1,\sdots,a_n \vdash \fork(b_1.\sdots\fork(b_p.\wait(u_{p+1},\sop),\wait(u_p,t_p)),\sdots\wait(u_1,t_1))
\end{displaymath}
The corresponding poset has $n$ minimal elements corresponding to the free thread IDs $a_1,\sdots,a_n$.
There is one labelled element for each of the terms $t_1$ to $t_p$.
The parent, which stops, corresponds to the distinguished maximal element $\possop$.
The partial order (solid line) encodes the dependencies given by the compound thread IDs $u_1$ to $u_{p+1}$.
The visibility relation (dotted line) corresponds to the thread IDs passed to a variable $t_i=x(u_{i1},\sdots,u_{im})$.
The closure conditions on normal forms from~\Cref{sec:clos-cond-norm} correspond to the transitivity of the partial order and to the fact that the visibility relation is downward-closed with respect to the partial order.


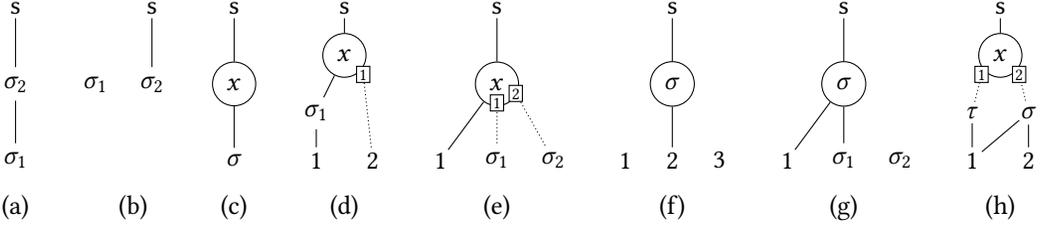
\begin{figure}
  \begin{minipage}[t]{1\linewidth}
  \begin{mathpar}
    \begin{tikzpicture}
	\begin{pgfonlayer}{nodelayer}
		\node [style=action] (0) at (0, 0) {$\possop$};
		\node [style=action] (1) at (0, -2) {$\sigma_2$};
		\node [style=action] (2) at (0, -4) {$\sigma_1$};
		\node [style=action] (3) at (0, -5.25) {(a)};
	\end{pgfonlayer}
	\begin{pgfonlayer}{edgelayer}
		\draw (1) to (0);
		\draw (2) to (1);
	\end{pgfonlayer}
\end{tikzpicture}
    \and
    \begin{tikzpicture}
	\begin{pgfonlayer}{nodelayer}
		\node [style=action] (0) at (0, 0) {$\possop$};
		\node [style=action] (1) at (0, -2) {$\sigma_2$};
		\node [style=action] (2) at (-1.5, -2) {$\sigma_1$};
		\node [style=action] (3) at (-0.5, -5.25) {(b)};
	\end{pgfonlayer}
	\begin{pgfonlayer}{edgelayer}
		\draw (1) to (0);
	\end{pgfonlayer}
\end{tikzpicture}
    \and
    \begin{tikzpicture}
	\begin{pgfonlayer}{nodelayer}
		\node [style=action] (1) at (0, 0) {$\possop$};
		\node [style=var] (2) at (0, -2) {$x$};
		\node [style=action] (3) at (0, -4) {$\sigma$};
		\node [style=action] (4) at (0, -5.25) {(c)};
	\end{pgfonlayer}
	\begin{pgfonlayer}{edgelayer}
		\draw (3) to (2);
		\draw (2) to (1);
	\end{pgfonlayer}
\end{tikzpicture}
    \and
    \begin{tikzpicture}
	\begin{pgfonlayer}{nodelayer}
		\node [style=action] (1) at (0, 0) {$\possop$};
		\node [style=var] (2) at (0, -1.25) {$x$};
		\node [style=action] (3) at (-0.75, -2.75) {$\sigma_1$};
		\node [style=input] (4) at (0.5, -1.75) {\tiny$1$};
		\node [style=action] (5) at (-0.75, -4) {$1$};
		\node [style=action] (6) at (0.75, -4) {$2$};
		\node [style=action] (7) at (0, -5.25) {(d)};
	\end{pgfonlayer}
	\begin{pgfonlayer}{edgelayer}
		\draw (3) to (2);
		\draw (2) to (1);
		\draw (5) to (3);
		\draw [style=visible] (6) to (4);
	\end{pgfonlayer}
\end{tikzpicture}
    \and
    \begin{tikzpicture}
	\begin{pgfonlayer}{nodelayer}
		\node [style=action] (1) at (0, 0) {$\possop$};
		\node [style=var] (2) at (0, -2) {$x$};
		\node [style=action] (3) at (0, -4) {$\sigma_1$};
		\node [style=input] (4) at (0, -2.5) {\tiny$1$};
		\node [style=action] (5) at (-1.5, -4) {$1$};
		\node [style=action] (6) at (1.5, -4) {$\sigma_2$};
		\node [style=input] (7) at (0.5, -2.25) {\tiny$2$};
		\node [style=action] (8) at (0, -5.25) {(e)};
	\end{pgfonlayer}
	\begin{pgfonlayer}{edgelayer}
		\draw (2) to (1);
		\draw [style=visible] (6) to (7);
		\draw [style=visible] (3) to (4);
		\draw (5) to (2);
	\end{pgfonlayer}
\end{tikzpicture}
    \and
    \begin{tikzpicture}
	\begin{pgfonlayer}{nodelayer}
		\node [style=action] (1) at (0, 0) {$\possop$};
		\node [style=var] (2) at (0, -2) {$\sigma$};
		\node [style=action] (3) at (0, -4) {$2$};
		\node [style=action] (5) at (-1.25, -4) {$1$};
		\node [style=action] (6) at (1.25, -4) {$3$};
		\node [style=action] (7) at (0, -5.25) {(f)};
	\end{pgfonlayer}
	\begin{pgfonlayer}{edgelayer}
		\draw (2) to (1);
		\draw (3) to (2);
	\end{pgfonlayer}
\end{tikzpicture}
    \and
    \begin{tikzpicture}
	\begin{pgfonlayer}{nodelayer}
		\node [style=action] (1) at (0, 0) {$\possop$};
		\node [style=var] (2) at (0, -2) {$\sigma$};
		\node [style=action] (3) at (0, -4) {$\sigma_1$};
		\node [style=action] (5) at (-1.5, -4) {$1$};
		\node [style=action] (6) at (1.5, -4) {$\sigma_2$};
		\node [style=action] (7) at (0, -5.25) {(g)};
	\end{pgfonlayer}
	\begin{pgfonlayer}{edgelayer}
		\draw (2) to (1);
		\draw (5) to (2);
		\draw (3) to (2);
	\end{pgfonlayer}
\end{tikzpicture}
    \and
    \begin{tikzpicture}
	\begin{pgfonlayer}{nodelayer}
		\node [style=action] (0) at (0, 0) {$\possop$};
		\node [style=var] (1) at (0, -1.25) {$x$};
		\node [style=action] (2) at (-0.75, -2.75) {$\tau$};
		\node [style=action] (3) at (0.75, -2.75) {$\sigma$};
		\node [style=input] (5) at (-0.5, -1.75) {\tiny$1$};
		\node [style=input] (6) at (0.5, -1.75) {\tiny$2$};
		\node [style=action] (7) at (-0.75, -4) {$1$};
		\node [style=action] (8) at (0.75, -4) {$2$};
		\node [style=action] (9) at (0, -5.25) {(h)};
	\end{pgfonlayer}
	\begin{pgfonlayer}{edgelayer}
		\draw (7) to (2);
		\draw (7) to (3);
		\draw (8) to (3);
		\draw (1) to (0);
		\draw [style=visible] (2) to (5);
		\draw [style=visible] (3) to (6);
	\end{pgfonlayer}
\end{tikzpicture}
  \end{mathpar}
  \end{minipage}
  \vspace{-1em}
  \caption{Examples of labelled posets.
    (a) is the term $\fork(a.\wait(a,\act{\sigma_2}),\act{\sigma_1})$.
    (b) is $\fork(a.\act{\sigma_2},\act{\sigma_1})$.
    (c) is $\fork(a.\wait(a,x),\act{\sigma})$.
    (d) is the normal form in~\cref{eq:nf-ex-2}.
    (g) is the result of substituting (f) for $(x:2)$ in $(e)$.
    (h) is the representation of~\Cref{fig:node-ops}~(c) using the notion of labelled poset introduced in this section.
  }
  \label{fig:lab-poset-ex-1}
\end{figure}

\subsubsection{Substitution for Labelled Posets}\label{sec:subst-labell-posets}
Just like terms in a parameterized theory, labelled posets admit substitution of another poset for a computation variable; the formal definition is discussed in~\Cref{sec:main-theorem}.
In~\Cref{fig:lab-poset-ex-1}, posets (e) and (f) represent respectively the terms
\begin{mathpar}
  x:2\vbar a_1 \vdash \fork(b_1.\fork(b_2.\wait(a_1,x(b_1,b_2)),\act{\sigma_2}),\act{\sigma_1})
  \and
  \text{and}
  \and
  \cdot\vbar a_1,b_1,b_2 \vdash \wait(b_1,\act{\sigma})
\end{mathpar}
while (g), the result of substituting (f) for $x:2$ in (e), is the term:
\begin{displaymath}
  \cdot\vbar a_1 \vdash \fork(b_1.\fork(b_2.\wait(a_1,\wait(b_1,\act{\sigma})),\act{\sigma_2}),\act{\sigma_1})
\end{displaymath}
The input $1$ of both (e) and (f) gets identified, while inputs $2$ and $3$ of (f) are mapped to the two inputs of variable $x$.

Substitution of a compound thread ID for an input of the poset corresponds to parameter substitution on terms.
We define this operation on posets in~\Cref{def:lab-poset-repr}.

\subsubsection{Connection to Labelled Posets from~\Cref{sec:model-caus-struct} and to Ordinary Labelled Posets}

The $\Sigma$-labelled posets with $n$ inputs and $m$ outputs from~\Cref{def:io-lab-poset} are a special case of the labelled posets from this section.
For example, the poset in~\Cref{fig:node-ops}~(c), which corresponds to the term $x:2 \vbar a_1,a_2 \vdash\node_\sigma(a_1\oplus a_2, b_1.\node_\tau(a_1, b_2.x(b_2, b_1)))$, is shown in~\Cref{fig:lab-poset-ex-1}~(h).
The two inputs of variable $x$ correspond to the two outputs of the poset from~\Cref{fig:node-ops}~(c).

If we regard the structure $\possop$ which marks the end of the main thread as itself a label, then a labelled poset with no inputs and no elements labelled by computation variables is an ordinary labelled poset in the sense of~\Cref{def:labelled-poset}, i.e.~a partially ordered set $(X,\leq)$ with a function $X\rightarrow\Sigma\uplus\{\possop\}$.

\subsection{Labelled Posets with Holes Formally}\label{sec:labell-posets-form}

The following two definitions formalize the ideas about labelled posets from the previous section.

\begin{definition}\label{def:repr-lab-poset}
  Let $\Gamma=x_1:m_1,\sdots,x_k:m_k$ be a context of computation variables, and $\Sigma$ be a set of observable actions.
  A \emph{$(\Gamma,\Sigma)$-labelled poset with $n$ inputs} $G=\langle V_1,V_2,\leq_G,l_1,l_2\rangle$ consists of:
  \begin{itemize}
  \item the set of $n$ input vertices $[n] = \{1,\dots,n\}$;
  \item finite disjoint sets of vertices $V_1$ (labelled by actions) and $V_2$ (labelled by variables);
\label{fig:lab-poset-not-well-formed}
  \item a distinguished vertex $\possop$;
  \item a partial order $\leq_G$ on the underlying set $\abs{G}\defeq [n] \uplus V_1\uplus V_2 \uplus \{\possop\}$ (depicted by solid lines);
  \item a labelling function $l_1:V_1\rightarrow \Sigma$, from the vertices in $V_1$ to observable actions;
  \item a function $l_2$ that labels the vertices in $V_2$ with variables $(x:m)$ from $\Gamma$:
    \begin{displaymath}
      l_2:V_2\rightarrow \sum_{(x:m)\in\Gamma}\big(f:[m]\rightarrow \MCP(\abs{G}) \big)
    \end{displaymath}
    and depending on the arity $m$ of this variable, $l_2$ also assigns a function $f:[m]\rightarrow \MCP(\abs{G})$ into the powerset of $\abs{G}$.
    (Each $f(i)$ is depicted by dotted lines).
  \end{itemize}
\end{definition}
If there are no inputs and no vertices labelled by variables, $n=0$ and $V_2=\emptyset$, then a labelled poset becomes an ordinary labelled poset with labels from $\Sigma\uplus\{\possop\}$.
An isomorphism $\alpha:G\rightarrow G'$ between labelled posets consists of two bijections $\alpha_1:V_1\rightarrow V'_1$ and $\alpha_2:V_2\rightarrow V'_2$ which preserve the two labelling functions, in a suitable sense, and such that  $\mathsf{id}_{[n]}\uplus\alpha_1\uplus\alpha_2\uplus \mathsf{id}_{\{\possop\}}$ preserves and reflects the partial order.


To obtain the correspondence between labelled posets and normal forms explained in~\Cref{sec:lab-poset-as-nf}, we restrict our attention to labelled posets that are well-formed.

\begin{definition}\label{def:lab-poset-well-formed}
  An $(\Gamma,\Sigma)$-labelled poset  with $n$ inputs $G=\langle V_1,V_2,\leq_G,l_1,l_2\rangle$ is \emph{well-formed} if:
  \begin{enumerate}
  \item the $n$ inputs are minimal and $\possop$ is maximal, with respect to the partial order $\leq_G$;
  \item for each $e\in V_2$ such that $l_2(e)=(x:m,\,f:[m]\rightarrow\MCP(\abs{G}))$ and for each $i\in [m]$: $e\in f(i)$ and $\possop\not\in f(i)$, and $f(i)$ is downward-closed with respect to $\leq_G$.
  \item\label{item:well-formed-pos-cond-3} Consider the \emph{visibility} relation $S\subseteq \abs{G}\times\abs{G}$ induced by the labelling function $l_2$:
    \begin{center}
      $(e',e)\in S$ $\Longleftrightarrow$ $e\in V_2$ and if $l_2(e)=(x:m,f)$ then $e'$ is in the image of $f$.
    \end{center}
    The transitive closure of the relation $({\leq_G})\cup S$ is anti-symmetric.
  \end{enumerate}
\end{definition}

All the labelled posets discussed in~\Cref{sec:labell-posets-exampl} satisfy the well-formedness conditions above.
Requiring that $s$ is maximal and $s\not\in f(i)$ ensures that child threads do not know the ID of the main thread.
Downward-closure of $f(i)$ and $e\in f(i)$ correspond to some of the closure conditions on normal forms~(\Cref{sec:clos-cond-norm}).
Condition~(\ref{item:well-formed-pos-cond-3}) ensures that, when taking into account the visibility relation, there are no cycles in the labelled poset.
Overall, well-formedness ensures that a labelled poset can be linearly ordered into a normal form of shape~(\ref{eq:normal-form}).
For example, the poset in~\Cref{fig:ill-formed-poset} cannot be linearized and does not satisfy condition~(\ref{item:well-formed-pos-cond-3}).

\subsection{The Labelled Poset Model and Representation Theorem}\label{sec:labelled-poset-model}

\begin{figure}
  \centering
  \begin{minipage}[b]{.15\linewidth}
    \begin{displaymath}
      \begin{tikzpicture}
	\begin{pgfonlayer}{nodelayer}
		\node [style=action] (0) at (0, 0) {$\possop$};
		\node [style=action] (1) at (0, -1.5) {$\sigma$};
		\node [style=var] (2) at (0, -3) {$x$};
		\node [style=input] (3) at (0.5, -3.5) {\tiny$1$};
	\end{pgfonlayer}
	\begin{pgfonlayer}{edgelayer}
		\draw (2) to (1);
		\draw [style=visible, in=-60, out=60, looseness=2.00] (1) to (3);
	\end{pgfonlayer}
\end{tikzpicture}
    \end{displaymath}
    \vspace{-2.5em}
    \caption{Example of an ill-formed poset.}
    \label{fig:ill-formed-poset}
  \end{minipage}
  \hfill
  \begin{minipage}[b]{0.82\linewidth}
  \begin{mathpar}
    \begin{tikzpicture}
	\begin{pgfonlayer}{nodelayer}
		\node [style=action] (0) at (0, 0) {};
		\node [style=action] (1) at (0, 0) {$\possop$};
		\node [style=action] (2) at (0, -2) {$\sigma_2$};
		\node [style=action] (3) at (-1.25, -2) {$\sigma_1$};
		\node [style=action] (4) at (-1.25, -4) {$1$};
		\node [style=action] (5) at (0, -4) {$2$};
		\node [style=var] (6) at (1.5, -2) {$x$};
		\node [style=input] (7) at (2, -2.5) {\tiny$1$};
		\node [style=action] (8) at (0, -5.25) {(a)};
	\end{pgfonlayer}
	\begin{pgfonlayer}{edgelayer}
		\draw (2) to (1);
		\draw (4) to (3);
		\draw (5) to (2);
		\draw [style=visible, in=-90, out=30, looseness=0.75] (5) to (7);
	\end{pgfonlayer}
\end{tikzpicture}
    \and
    \begin{tikzpicture}
	\begin{pgfonlayer}{nodelayer}
		\node [style=action] (0) at (0, 0) {};
		\node [style=action] (1) at (0, 0) {$\possop$};
		\node [style=action] (2) at (0, -2) {$\sigma_3$};
		\node [style=action] (5) at (0, -4) {$1$};
		\node [style=action] (6) at (0, -5.25) {(b)};
	\end{pgfonlayer}
	\begin{pgfonlayer}{edgelayer}
		\draw (2) to (1);
		\draw (5) to (2);
	\end{pgfonlayer}
\end{tikzpicture}
    \and
    \begin{tikzpicture}
	\begin{pgfonlayer}{nodelayer}
		\node [style=action] (0) at (0, 0) {};
		\node [style=action] (1) at (0, 0) {$\possop$};
		\node [style=action] (2) at (0, -2) {$\sigma_3$};
		\node [style=action] (5) at (0, -4) {$1$};
		\node [style=action] (6) at (0, -5.25) {(c)};
	\end{pgfonlayer}
	\begin{pgfonlayer}{edgelayer}
		\draw (5) to (2);
	\end{pgfonlayer}
\end{tikzpicture}
    \and
    \begin{tikzpicture}
	\begin{pgfonlayer}{nodelayer}
		\node [style=action] (0) at (0, 0) {};
		\node [style=action] (1) at (0, 0) {$\possop$};
		\node [style=action] (2) at (0, -1.25) {$\sigma_2$};
		\node [style=action] (3) at (-1.5, -2) {$\sigma_1$};
		\node [style=action] (4) at (0, -4) {$1$};
		\node [style=action] (5) at (0, -2.75) {$\sigma_3$};
		\node [style=var] (6) at (1.5, -2) {$x$};
		\node [style=input] (7) at (2, -2.5) {\tiny$1$};
		\node [style=action] (8) at (0, -5.25) {(d)};
	\end{pgfonlayer}
	\begin{pgfonlayer}{edgelayer}
		\draw (2) to (1);
		\draw (4) to (3);
		\draw (5) to (2);
		\draw [style=visible, in=-90, out=0, looseness=1.25] (5) to (7);
		\draw (4) to (5);
	\end{pgfonlayer}
\end{tikzpicture}
    \and
    \begin{tikzpicture}
	\begin{pgfonlayer}{nodelayer}
		\node [style=action] (0) at (0, 0) {};
		\node [style=action] (1) at (0, 0) {$\possop$};
		\node [style=action] (2) at (0, -2) {$\sigma_2$};
		\node [style=action] (3) at (-2.5, -2) {$\sigma_1$};
		\node [style=action] (4) at (-1.25, -4) {$1$};
		\node [style=var] (6) at (1.25, -2) {$x$};
		\node [style=input] (7) at (1.75, -2.5) {\tiny$1$};
		\node [style=action] (8) at (-1.25, -2) {$\sigma_3$};
		\node [style=action] (9) at (-0.25, -5.25) {(e)};
	\end{pgfonlayer}
	\begin{pgfonlayer}{edgelayer}
		\draw (2) to (1);
		\draw (4) to (3);
		\draw (4) to (8);
	\end{pgfonlayer}
\end{tikzpicture}
    \and
    \begin{tikzpicture}
	\begin{pgfonlayer}{nodelayer}
		\node [style=action] (1) at (1.25, 0) {$\possop$};
		\node [style=action] (2) at (0, -2) {$\sigma_3$};
		\node [style=action] (5) at (0, -4) {$1$};
		\node [style=action] (6) at (0.5, -5.25) {(f)};
		\node [style=action] (7) at (1.25, -4) {$2$};
	\end{pgfonlayer}
	\begin{pgfonlayer}{edgelayer}
		\draw (5) to (2);
		\draw (7) to (2);
		\draw (7) to (1);
	\end{pgfonlayer}
\end{tikzpicture}
  \end{mathpar}
  \vspace{-1em}
  \caption{Examples of forking on labelled posets:
    forking parent (a) with child (b) gives (d).
    If instead, the child is (c), the result of forking is (e).
    (f) is the result of applying $\wait_1$ to (c).}
  \label{fig:fork-examples}
  \end{minipage}
\end{figure}

\subsubsection{Model Structure}

In order to build a model for the theory of dynamic threads out of labelled posets, we organize them into a functor $\repr{\Gamma}:\indexcat\rightarrow \Set$ as follows.

\begin{definition}\label{def:lab-poset-repr}
  Let $\Gamma$ be a context of computation variables.
  For each natural number $n$, define the set $\repr{\Gamma}(n)$ to contain \emph{isomorphism classes of well-formed $(\Gamma,\Sigma)$-labelled posets with $n$ inputs}.
  The functorial action on relations $R\subseteq[n]\times[n']$ acts on the inputs of a poset by updating $i \leq e$ to $i' \leq e$ if $(i,i') \in R$, and updating the labelling function $l_2$ accordingly.
\end{definition}

We can equip $\repr{\Gamma}$ with a model structure, in the sense of~\Cref{def:sigma-struct}, for the $\fork$, $\wait$, $\sop$ and $\act{\sigma}$ operations of the theory of dynamic threads.
For each natural number $n$, we define:
  \begin{displaymath}
    \sfork_n : \repr{\Gamma}(n+1)\times \repr{\Gamma}(n) \rightarrow \repr{\Gamma}(n).
  \end{displaymath}
The labelled vertices of $\sfork_n(G_1,G_2)$ are the union of those from $G_1$ and $G_2$ and the labels are preserved.
The partial order of $\sfork_n(G_1,G_2)$ is obtained by connecting the $(n+1)$-th input of $G_1$ to the $\possop$ element of $G_2$ and closing under transitivity.
The visibility relation is obtained via the same connection from those of $G_1$ and $G_2$.

\Cref{fig:fork-examples} shows an example of forking.
The parent (a) and the children (b) and (c) correspond respectively to the terms:
\begin{mathpar}
  x:1\vbar a_1,a_2 \vdash t_1 = \fork(b_1.\fork(b_2.\wait(a_2,\sigma_2),x(a_2)),\wait(a_1,\sigma_1))
  \and
  x:1\vbar a_1 \vdash t_2 = \wait(a_1,\act{\sigma_3})
  \and
  x:1\vbar a_1 \vdash t_3 = \fork(c.\sop, \wait(a_1,\act{\sigma_3}))
\end{mathpar}
while (d) corresponds to $\fork(a_2.t_1,t_2)$ and (e) corresponds to $\fork(a_2.t_1,t_3)$.


The operation $\swait_n: \repr{\Gamma}(n) \rightarrow \repr{\Gamma}(n+1)$ adds a new input $n+1$ and connects it to all the labelled elements and to $\possop$.
\Cref{fig:fork-examples}~(f) is an example; it represents the term $\cdot\vbar a_1,a_2\vdash \wait(a_2,t_3)$.

The operation $\ssop_n:1\rightarrow \repr{\Gamma}(n)$ picks out the poset with only the $n$ inputs and $\possop$ as elements, no labelled vertices, and with the discrete partial order.
The operation $\sprint{\sigma,n}:1\rightarrow \repr{\Gamma}(n)$ gives a poset with one vertex labelled by $\sigma$, directly below $\possop$ in the partial order, and with the $n$ inputs not connected to anything.

\subsubsection{Main Theorem}\label{sec:main-theorem}
We now show that, for each $T_\Gamma$, labelled posets form a model for the theory of dynamic threads and that this model is free.
Recall that we described the term model $F_\forkth(V_\Gamma)$  as a free model on $V_\Gamma$ in~\Cref{sec:free-models}.
In~\Cref{sec:interpretation}, we give an interpretation of our programming language using the fact that $T_\Gamma:\FinRel\to\Set$ induces a strong monad $T$ on the functor category $\Set^\FinRel$ by setting $T(V_\Gamma)=T_\Gamma$.

  \begin{theorem}[Representation Theorem]\label{thm:main-theorem}
    For each context $\Gamma$, the functor $\repr{\Gamma}$ from~\Cref{def:lab-poset-repr} and the operations on labelled posets $\fork$, $\wait$, $\sop$, $\act{\sigma}$ respect the equations of the theory of dynamic threads from~\Cref{fig:fork-wait-eqs}, and thus form a model of the theory, in the sense of~\Cref{def:param-model}.

     Moreover, $\repr{\Gamma}$ is isomorphic to the term model $F_\forkth(V_\Gamma)$ of the theory of dynamic threads.
  \end{theorem}

  Recall that in~\Cref{sec:subst-labell-posets} we informally discussed substitution for labelled posets.
  Using the isomorphism of models in the theorem above we can define substitution by translating a labelled poset into an equivalence class of terms, using term substitution, then translating back to a poset.

  \subsubsection{Proof Sketch of the Representation Theorem}\label{sec:proof-sketch-repr}

  To prove~\Cref{thm:main-theorem} we consider the diagram below.
  Recall that the functor $\nfset{\Gamma}$, from~\Cref{def:normal-form}, contains equivalence classes of normal forms, and $F_\forkth(V_\Gamma)$ contains equivalence classes of terms.

  \begin{wrapfigure}[5]{l}{0.2\textwidth}
    \vspace{-2mm}
\begin{tikzcd}
F_\forkth(V_\Gamma) \arrow[r, "\interp"]  & \repr{\Gamma} \arrow[ld, "\reify"] \\
\nfset{\Gamma} \arrow[u, "\inc"] &
\end{tikzcd}
  \end{wrapfigure}
  The $\inc$ map is given by the inclusion of normal forms into terms.
  The map $\interp$ is the unique map obtained by instantiating the universal property of the free model $F_\forkth(V_\Gamma)$~(\Cref{def:free-model}) for a suitable $V_\Gamma\rightarrow \repr{\Gamma}$ that maps computation variables to labelled posets; $\interp$ is essentially given by the interpretation $\den{-}_{\repr{\Gamma}}$ of terms in the labelled poset model, from~\Cref{sec:models-param-theor}.

  For each natural number $n$, we define a function $\reify_n$ that linearizes a labelled poset into a normal form, using the intuition from~\Cref{sec:lab-poset-as-nf}.
  To show that this gives a well-defined function, natural in $n$, we use the conditions for a well-formed labelled poset~(\Cref{def:lab-poset-well-formed}), the closure conditions and the equivalence relation on normal forms~(\Cref{sec:clos-cond-norm}).

  We show that $\reify$ is both a left and a right inverse to the composite $\interp\circ\inc$.
  To show it is a left inverse we use induction on the number of child threads in a normal form; for the right inverse, we use induction on the number of labelled elements in a poset.
  Knowing that $\inc$ is surjective~(\Cref{thm:normal-form-surj}) means that $\interp$ is an isomorphism.
  We already know from the definition of the free model that $\interp$ is a homomorphism of models.

  \begin{corollary}\label{cor:norm-form-unique}
    The inclusion of equivalence classes of normal forms into equivalence classes of terms is injective.
    By~\Cref{thm:normal-form-surj}, every term is equal to a unique equivalence class of normal forms.
  \end{corollary}

\section{A Completeness Theorem for the Theory of Dynamic Threads}\label{sec:compl-theor-theory}

\Cref{thm:main-theorem} shows that terms $\Gamma\vbar a_1,\sdots,a_n \vdash t$ in the theory of dynamic threads correspond exactly to $(\Gamma,\Sigma)$-labelled posets with $n$ inputs.
When the context $\Gamma$ is empty and $n=0$, these labelled posets are in fact ordinary labelled posets i.e.~a partially ordered set $(X,\leq)$ with a function $X\rightarrow\Sigma\uplus\{\possop\}$.
The next theorem shows that our axiomatization of $(\Gamma,\Sigma)$-labelled posets from~\Cref{fig:fork-wait-eqs} is complete with respect to an equivalence relation induced by isomorphism of ordinary labelled posets.

Given a term in context $\Gamma\vbar \Delta\vdash t$, a closing substitution $\gamma$ is one that assigns to each variable $(x:m)$ from $\Gamma$ a term $\gamma(x)=\bigl(\cdot\vbar\Delta,b_1,\sdots,b_m \vdash s \bigr)$ with no free computation variables, so that $\cdot\vbar \Delta \vdash t[\gamma]$ holds.
A closing context $\MCC[-]$ is a term with a hole such that given a term $\cdot\vbar \Delta \vdash t$, the judgement $\cdot\vbar \cdot \vdash \MCC[t]$ holds.

\begin{theorem}[Completeness]\label{thm:completeness}
  Consider terms $\Gamma \mid \Delta \vdash t_1,t_2$ in the theory of dynamic threads.
  If for all closing substitutions $\gamma$ and for all closing contexts $\MCC[-]$ the two terms are equal, meaning $\cdot\vbar \cdot \vdash \MCC[t_1[\gamma]] = \MCC[t_2[\gamma]]$, then $\Gamma \mid \Delta \vdash t_1=t_2$.
\end{theorem}

\begin{proof}[Proof sketch]
  Consider terms with no free computation variables $\cdot \mid a_1,\sdots,a_n \vdash t_1,t_2$.
  Define a context which binds each free thread ID to an observable action, and adds an  action $\act{\sigma_{n+1}}$ at the end of the main thread ($\sigma_{1},\sdots,\sigma_{n+1}$ are distinct from any observable actions occurring in $t_1$ and $t_2$):
  \begin{displaymath}
    \MCC[-] = \fork(a_1.\sdots\fork(a_n.\fork(a_{n+1}.\wait(a_{n+1},\act{\sigma_{n+1}}),[-]),\act{\sigma_{n}}),\sdots \act{\sigma_{1}})
  \end{displaymath}

  Recall that $\MCC[t_1]$ and $\MCC[t_2]$ are interpreted as labelled posets in $\repr{\emptyset}(0)$.
  If they are equal then by~\Cref{thm:main-theorem} there is an isomorphism of labelled posets between $\den{\MCC[t_1]}$ and $\den{\MCC[t_1]}$.
  We adapt this isomorphism into one between $\den{t_1},\den{t_2}\in \repr{\emptyset}(n)$ and deduce $t_1=t_2$, but we omit the details.

\noindent\begin{tabular}{@{}p{.72\textwidth}p{.27\textwidth}}
\hspace{\parindent}
Now consider terms in context $\Gamma \mid a_1,\sdots,a_n \vdash t_1,t_2$.
Consider the substitution $\gamma$ which maps each computation
variable $(x:m)$ from $\Gamma$ to the term in context $\cdot\vbar
a_1,\sdots,a_n,b_1,\sdots,b_{m}$ depicted on the right as
a labelled poset with $n+m$ inputs.
  We assume that the actions $\sigma_{x},\sigma_{x_1},\sdots,\sigma_{x_{m}}$ are distinct for each computation variable $x$, and distinct from the actions in $t_1$ and $t_2$.
  We may encode ``new'' actions as distinct combinations. Even for a single action, this can be done by alternating differing parallel combinations.
  &
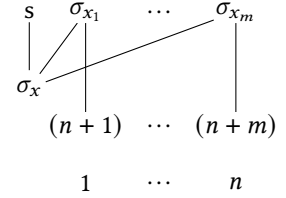
  \begin{tikzpicture}
	\begin{pgfonlayer}{nodelayer}
		\node [style=action] (1) at (0, 0) {$\possop$};
		\node [style=medium box] (2) at (1.5, 0) {$\sigma_{x_1}$};
		\node [style=action] (3) at (3.5, 0) {\dots};
		\node [style=medium box] (4) at (5.5, 0) {$\sigma_{x_{m}}$};
		\node [style=action] (5) at (0, -2) {$\sigma_x$};
		\node [style=medium box] (6) at (1.5, -3) {$(n+1)$};
		\node [style=medium box] (7) at (5.5, -3) {$(n+m)$};
		\node [style=action] (8) at (5.5, -4.5) {$n$};
		\node [style=action] (9) at (3.5, -4.5) {\dots};
		\node [style=action] (10) at (1.5, -4.5) {$1$};
		\node [style=action] (11) at (3.5, -3) {\dots};
	\end{pgfonlayer}
	\begin{pgfonlayer}{edgelayer}
		\draw (5) to (1);
		\draw (5) to (2);
		\draw (5) to (4);
		\draw (6) to (2);
		\draw (7) to (4);
	\end{pgfonlayer}
\end{tikzpicture}
\end{tabular}

  Assuming we have an isomorphism of labelled posets $\alpha:\den{t_1[\gamma]}\rightarrow \den{t_2[\gamma]}$, we can construct an isomorphism $\alpha':\den{t_1}\rightarrow \den{t_2}$.
  On labelled elements, $\alpha'$ acts the same as $\alpha$, forgetting about the elements labelled $\sigma_{x_1},\sdots,\sigma_{x_{m}}$.
  The elements labelled $\sigma_x$ in $\den{t_1[\gamma]}$ correspond exactly to the elements labelled $(x:m)$ in $\den{t_1}$, and the dependencies of the elements $\sigma_{x_i}$ encode the visibility relation of $\den{t_1}$.
  To show $\alpha'$ is an isomorphism of $(\Gamma,\Sigma)$-labelled posets with $n$ inputs, we use the well-formedness of labelled posets that represent terms.
  In particular, we rely on condition~(\ref{item:well-formed-pos-cond-3}) from~\Cref{def:lab-poset-well-formed} which says that the visibility relation does not induce cycles.
\end{proof}

\section{Denotational Semantics, Soundness, Adequacy and Full Abstraction}\label{sec:denot-semant-progr}

\Crefrange{sec:param-algebr-theor}{sec:compl-theor-theory} have developed a theory for an algebraic language based on fork and wait.
The idea of algebraic effects is that this can easily be extended to a full language.
To demonstrate this, we return to the programming language from Section~\ref{sec:background}, outlining how it can be given a semantics by using our complete representation, which is sound, adequate, and fully abstract at first order with respect to the operational semantics.

\subsection{Interpretation}\label{sec:interpretation}
\subsubsection{Summary}We interpret all types $A$ of the programming language as functors $\den A\in\Set^\FinRel$.
The idea is that $\den A(w)$ is the set of interpretations of values of type $A$ in world $w$, i.e.~%
$\vterm w v:A$.
Similarly, we interpret contexts $\Gamma$ as functors $\den \Gamma\in\Set^\FinRel$, i.e.~world-indexed sets of valuations.

We will interpret value expressions $\Gamma\vterm{w} v:A$ in world $w$ as natural transformations $\den v: \FinRel(w,-) \times \den\Gamma\to\den A\text,$ in $\Set^\FinRel$, 
where $\FinRel(w,-)$ is the representable functor at $w$.
If the world $w$ is empty this reduces to a natural transformation $\den\Gamma\to\den A$.
To interpret computation expressions, we build a monad $T$ on $\Set^\FinRel$ from the representation, and interpret expressions $\Gamma \cterm{w} t: A$ as
natural transformations
$\den t:\FinRel(w,-) \times\den\Gamma\to T\den A$.

\subsubsection{Interpretation of First Order Types}
The interpretation of the type of thread IDs is:
\[
  \den \tid= \FinRel(1,-)\quad i.e.~\den\tid(w)\cong \{w'~|~w'\subseteq w\}\text.
\]
We thus interpret (compound) thread id values in world $w$ as subsets of (non-compound) tids in $w$.

The interpretation of product and sum types uses the well-known and canonical categorical structure of the functor category $\Set^\FinRel$. Recall that products and coproducts in functor categories are computed pointwise.
Moreover, we have a strong connection with Section~\ref{sec:free-models},
since
for a context $(x_1:m_1\dots x_k:m_k)$, the functor of variables is an interpretation
of a type:
\begin{equation}\label{eqn:vartype}
\textstyle\hspace{-3mm}
  \den {\prod_{i=1}^k A_i} (p)=
  \prod_{i=1}^k \den {A_i}(p)
  \quad
  \den {\sum_{i=1}^k A_i} (p)=
  \biguplus_{i=1}^k \den {A_i}(p)
\quad \textstyle  V_{x_1:m_1\dots x_k:m_k}\cong \den {\sum_{i=1}^k \tid^{m_i}}\text.
\end{equation}
In fact, every first-order type is isomorphic to one of this form, since products distribute over sums.

\subsubsection{Monad}

We extend the parameterized algebraic theory for fork and wait to a strong monad on $\Set^\FinRel$, along the lines of~\cite{Staton13,Staton13PL}.
Recall~(e.g.~\cite{PlotkinP02}) that a plain algebraic theory (such as monoids) induces a monad (such as lists) on the category of sets, by letting $T(X)$ be the free model of the theory on the set $X$.
For a parameterized algebraic theory (such as the theory of fork and wait), we can define a strong monad $T$ on the category $\Set^\FinRel$ by letting $T(X)$ be the free model of the theory on the functor $X\in \Set^\FinRel$. Recall that in~\Cref{sec:labelled-poset-model}, we built a model of the theory of dynamic threads, $T_\Gamma$, for each context $\Gamma$, and showed that it is the free model over $V_\Gamma$~(\Cref{thm:main-theorem}).
Thus, for first-order types, which are all interpreted as some $V_\Gamma$, we let $T(V_\Gamma)$ be the functor $T_\Gamma$.

Aside: every strong monad on $\Set^\FinRel$ that preserves sifted colimits arises from a parameterized algebraic theory in this way, giving a monad-theory correspondence.
To prove this correspondence one can generalize the results in~\cite[\S 5]{Staton13} straightforwardly, from a parameterizing Lawvere theory generated from a signature with no equations to an arbitrary Lawvere theory, in our case $\FinRel^\mathsf{op}$.

\subsubsection{Interpretation of Higher Order Types}

It is well-known that the category $\Set^\FinRel$ is cartesian closed. For functors
$G,H$ we have a functor $H^G$
determined by the currying isomorphism: to give a natural transformation $F\times G\to H$
is to give a natural transformation $F\to H^G$.
We then interpret the function type $A\to B$ using the monad, and this cartesian closed structure:
\[
  \den{A\to B}=
  (T\den B)^{\den A}\text.
\]
\subsubsection{Interpretation of Terms}
We interpret the concurrency-specific primitives ($\gfork,\gwait,\gsop,\gprintstop{}$) using the
fact that $T(X)$ is always a model of the parameterized algebraic
theory, as follows. These maps are sometimes called the `generic effects' of the algebraic operations~\cite{pp-algop-geneff}.
\begin{align*}
&  \den{\gfork}=\lambda (). \fork(\lambda a.\eta(\inj 1(a)),\eta(\inj 2()))
  :1\to T(\den\tid+1)&&  \den{\gsop}=\lambda (). \sop:1\to T(0)
  \\&
  \den{\gwait}=\lambda a. \wait(a,\eta()):\den\tid\to T(1)
&&  \den{\gprintstop\sigma}=\lambda (). \print\sigma:1\to T(0)
\end{align*}
  If we elide the third equation in \eqref{eqn:vartype} we can also regard these as the canonical terms:
  \[
    \gfork()=\fork(a.x(a),y())\qquad
    \gwait(a)=\wait(a,x)\qquad
    \gsop()=\sop\qquad
    \gprintstop\sigma()=\print\sigma
  \]
  The remainder of the interpretation of value and computation terms is the long-established interpretation of a call-by-value
  language in a bicartesian closed category with a monad~\cite{moggi_notions_1991}. The language constructs (sums, products and functions) match up
  with the categorical structure (coproducts, products, and cartesian closure).

  The interpretation of $\letin x \dots\dots$ is given using the monad strength and the multiplication.
  For first-order types, this amounts to the substitution of the parameterized algebraic theory. This is informative to spell out.
  Let $A=\sum_{i=1}^k \tid^{m_i} $ and $B=\sum_{i=1}^{k'} \tid^{m'_i}$. Consider program expressions:
  \begin{mathpar}
    \cterm{a_1\dots a_p} t : A
    \and
    x:A\cterm{a_1\dots a_p} u : B
  \end{mathpar}
  and we explain $\den{\letin x t u}$.
  For $1\leq i\leq k$, let $\cterm{a_1\dots a_p, c_1\dots c_{m_i}} u_i\defeq u[^{\inj i(\vec c)}\!/\!_x]:B$, where each $\inj i(\vec c)$ has type $\sum_{i=1}^k \tid^{m_i}$.
  Then, by the third clause in~\cref{eqn:vartype} and~\Cref{thm:main-theorem}, we can regard $\den t\in T(V_{x:m_1\dots x_k:m_k})(\vec a)$ and each $\den {u_i}\in T(V_{x'_1:m'_1\dots x'_k:m'_{k'}})(\vec a,\vec c)$ as terms
    \begin{mathpar}
      x_1:m_1,\sdots, x_k :m_k\vbar a_1,\sdots, a_p\vdash  \bar t
      \and
      x'_1:m'_1,\sdots, x'_{k'}:m'_{k'}\vbar a_1,\sdots, a_p , c_1,\sdots, c_{m_i}\vdash \bar u_i
    \end{mathpar}  
    in the parameterized algebraic theory. Then
    the interpretation of $\cterm{a_1\dots a_p}\letin x t u:B$ in $T(V_{m'_1\dots m'_{k'}})(\vec a)$ amounts to the following substituted term in the parameterized algebraic theory:
    \[
     x'_1:m'_1,\sdots, x'_k :m'_{k'}\vbar a_1,\sdots, a_p \vdash \bar t[^{\bar u_i}\!/\!_{x_i}]
    \]
    For example,
    $\den{\gprint\sigma()}\in T(1)(0)$ is the semantics of both
    the program for $\gprint \sigma()$ on the left,
    and the term in the parameterized algebraic theory on the right:
    \[\begin{array}{@{}l@{}}
    \cterm{\emptyset}\letin{z}{\gfork()}{
      \casevert{z}{
        \begin{array}{@{}l@{}}
          \clause1{a}{\gwait(a)}\\
          \clause2{}{\gprintstop \sigma()}
        \end{array}}
      :\tone}\qquad\!
    x:0\vbar \cdot \vdash \fork(a.\wait(a,x),\print \sigma)\text.
    \end{array}
    \]

\subsection{Adequacy, Contextual Equivalence, Soundness, and Full Abstraction}\label{sec:adeq-cont-equiv}

We prove a version of adequacy, which usually says that, at ground types, if the denotation of a term is equal to that of a value then the term reduces to that value.
We use adequacy to show that denotational equality is a sound proof technique for contextual equivalence.
For terms of first-order type, we show the converse (a partial full abstraction result), using the completeness result from~\Cref{thm:completeness}.
General adequacy results for algebraic effects have been proved by Plotkin and Power~\cite{DBLP:conf/fossacs/PlotkinP01} and Kavvos~\cite{DBLP:journals/pacmpl/Kavvos25} but these results do not apply to our setting directly because we express effects using a parameterized signature and model them in a functor category.

\begin{lemma}[Adequacy]\label{lemma:adequacy}
  For all $\cterm{\emptyset} t :\tzero$, we have
    $\sconfig t \Downarrow \den t$.
  \end{lemma}
  \begin{proof}[Proof outline]
      We prove this in three steps.
  \begin{enumerate}
  \item We extend term interpretations $\den t$ to well-formed configurations~$\den C$.
  \item
  We show a soundness property for the reduction relation: semantic interpretation is preserved by reduction.
  For example, if $C\longrightarrow C'$ then $\den C=\den {C'}$.
  This is a straightforward induction proof, but the statement is subtle, requiring accumulating the action labels.
\item
  We pick a reduction sequence from $\sconfig t$, noting by Proposition~\ref{prop:determinacy} that the choice of sequence doesn't matter and that it will terminate.
  A finished configuration only has the waiting relation $\preceq$ remaining, and all the stopped threads. With the accumulated action labels,
  this labelled poset is $\den t$, because reduction preserves semantic interpretations.
\end{enumerate}
 \vspace{-12pt}
\end{proof}
\begin{definition}
  Let $\Gamma$ be a typing context and $A$ a type.
  A \emph{program context} $\ctx-$ for $\Gamma,A$ is
  a program of type $\tzero$ with a hole of type $A$. Thus,
  if
  $\Gamma \cterm{\emptyset} t:A$ then $\cterm{\emptyset} \ctx t: \tzero$.

  Two programs $\Gamma \cterm{\emptyset} t,u:A$ are \emph{contextually equivalent},
  written \(
    t\ctxeq u
  \), if
  for every $(\Gamma,A)$ context $\ctx-$,
  letting $(\cong)$  denote isomorphism of labelled posets,
  we have that
  \[
    \sconfig {\ctx t}\Downarrow (P,\ell_P)    \  \ \& \ \     \sconfig {\ctx u}\Downarrow (Q,\ell_Q)
    \implies (P,\ell_P)\cong (Q,\ell_Q) 
  \]
\end{definition}

By Proposition~\ref{prop:determinacy}, the $(P,\ell_P)$ and $(Q,\ell_Q)$ are uniquely determined by $\ctx t$ and $\ctx u$ respectively.

\begin{proposition}[Soundness]
  Suppose that $\Gamma\cterm{\emptyset} t,u :A$.
If $\den t=\den u$ then
$    t\ctxeq u$.\label{thm:soundness}
\end{proposition}
\begin{proof}
  We deduce the result from Lemma~\ref{lemma:adequacy} as follows.
  We consider any two terms in any typing context, $\Gamma \cterm{\emptyset} t,u:A$, and any $(\Gamma,A)$-context, $\ctx-$.
  Suppose that $\den t=\den u$.
  From Lemma~\ref{lemma:adequacy}, $\sconfig {\ctx t}\Downarrow \den{\ctx t}$ and also
  $\sconfig {\ctx u}\Downarrow \den {\ctx u}$.
  Since the denotational semantics is compositional and $\den t=\den u$, also $\den{\ctx t}=\den{\ctx u}$.
  Thus $t\ctxeq u$.
\end{proof}
In particular, the standard $\beta/\eta$ laws are sound,
as are all the equations in Figure~\ref{fig:fork-wait-eqs},
such as~(\ref{eq:fork-wait-comm}):
\[
  \gwait(b);\gfork() \ =\
  \letin x {\gfork()} {\gwait(b);\ret x}\text.
\]
\begin{theorem}[Full abstraction at first order]\label{thm:full-abstraction}
  Suppose that $a_1\colon A_1,\sdots, a_p\colon A_p\cterm{\emptyset} t,u :B$
  and $A_1\dots A_p$ and $B$ are all first order (no function types).
  Then \(\den t=\den u\)
  if and only if \(
    t\ctxeq u\).
\end{theorem}
\begin{proof}
  From left to right follows from Theorem~\ref{thm:soundness}.
  From right to left, we first consider the case where $p=0$ and $B=\tzero$. Then
  contextual equivalence with the empty context in particular, together with Lemma~\ref{lemma:adequacy}, gives
  $\den t= \den u$.

  We next consider the case where $A_1=A_2=\dots A_p=\tid$ and $B=\sum_{i=1}^k \tid^{m_i}$.
  Suppose $t\ctxeq u$.
  Via~(\ref{eqn:vartype}), $\den t,\den u\in T(V_{x_1:m_1\dots x_k:m_k})(p)$, that is, $t$ and $u$ are interpreted directly in the parameterized algebraic theory.
  We must show that they are equal. By Theorem~\ref{thm:completeness},
  it suffices to show that $\ctx{\den t}[\gamma]=\ctx{\den u}[\gamma]$ for all algebraic contexts $\ctx-$ and algebraic substitutions $\gamma$.
  We deduce this by converting $\ctx-$ and $\gamma$ into a program context (`full definability' at first order) so that we can use the contextual equivalence $t\ctxeq u$.

  First, we note that the programming language supports algebraic operations at all types, via the generic effects:
    $\print\sigma =\gprintstop\sigma()$,
    $\sop=\gsop()$ and
  \begin{equation}\label{eqn:alg-to-geneff}
    \fork(t,u)=\casetwo{\gfork()}{a}t{()}u
    \qquad
    \wait(a,t)=\gwait(a);t
  \end{equation}
  We use this to convert the algebraic context $\ctx-$ to a program context that binds the free variables $a_1\dots a_p$. Moreover, each `substituend'
  $\gamma(x_i)$
  has no computation variables, and hence can also be regarded as a program of type~$\tzero$ under~\eqref{eqn:alg-to-geneff}.
  Now we define the computation program expression
  \[t[\gamma]\defeq
    \casesum{t} {\vec a}{\gamma(x_i)}{i}{1}{k}
  \]
  so that $\den {t[\gamma]}=\den{t}[\gamma]\in T(0)(p)$.
  Thus $\ctx{\den t}[\gamma]=\den{\ctx t [\gamma]}=\den{\ctx u [\gamma]}=\ctx {\den u}[\gamma]$ as required.
  Finally, we deduce the full result by using $\beta/\eta$ laws for sums and the fact that every first-order type is definably isomorphic to one of the form~(\ref{eqn:vartype}).
    \end{proof}

\section{Further Related and Future Work and Concluding Remarks}
    \label{sec:conc-etc}
Before concluding, we discuss some additional related work and
future directions our work enables.
\subsection{Further Related Work}

\paragraph{Algebraic Effects for Concurrency}
As briefly discussed in Section~\ref{sec:intro}, algebraic theories have been used to axiomatize features of process calculi, including in the style of algebraic effects.
This includes an algebraic-effects analysis of name creation and communication of names over channels in the $\pi$-calculus~\cite{DBLP:journals/tcs/Stark08}, and
a treatment of features of CSP such as action, choice and concealment~\cite{DBLP:books/daglib/p/GlabbeekP10} using algebraic effects and handlers.
From a programming language perspective, concurrency in the presence of nondeterminism and global shared state has been modelled using algebraic effects by Abadi and Plotkin~\cite{coop} and Dvir et al.~\cite{ra-conf, two-sorted-brookes}.
As discussed in~\Cref{sec:fork-wait-general},
our work differs from this previous work in that parallel composition of programs (i.e.~forking) is an operation in the equational axiomatization, whereas in previous work it was defined on top of the algebraic effects presentation.
The key ingredient that makes this possible is that we treat thread IDs as primitive and use the framework of parameterized algebraic theories to capture thread creation.

\paragraph{Trace Semantics}
Brookes's influential work~\cite{brookes:fa-shared-state} models a
preemptive concurrent programming language with global shared
state. Programs denote closed sets of traces; these traces represent a protocol involving the changes to memory by the program and its environment.
This form of semantics is robust under variation and
extensions~\cite{bhn,tw,xrh}, including variations to weak memory
models~\cite{tso-traces,ra-conf}.
Dvir et al.~\cite{two-sorted-brookes} give
a \emph{two-sorted} algebraic theory for Brookes-like traces.
Their representation theorem recovers Brookes's monad
when restricted to one of the sorts. Interestingly, the same
representation recovers Abadi and Plotkin's~\cite{coop} monad for
\emph{cooperative} concurrent programming with shared state when
restricted to the other sort.
Both Dvir et al.'s and Abadi and Plotkin's presentations presuppose non-deterministic choice as an algebraic
operation. In contrast, in our parameterized algebraic theory
the non-deterministic behaviour emerges from the more primitive
behavior of thread forking.

\paragraph{Effect Handlers for Concurrency}
Effect handlers arose from the theoretical study of algebraic
effects~\citep{PlotkinP13} as a way of supporting non-algebraic
effects, such as an operation for catching exceptions.
They were
quickly adopted as a general feature for modular programming
with effects~\citep{KammarLO13}, and are central to how
concurrency is currently implemented in OCaml
5~\citep{SivaramakrishnanDWKJM21} and in
WebAssembly~\citep{PhippsCostinRGLHSPL23}. They provide the basis for
a whole range of different concurrency effects such as actors,
async/await, coroutines, generators, and green threads.
Alas, the practice of programming with effect handlers departs
substantially from the established theory:
we do not yet know how to specify the semantics
of these effects using any kind of equational axiomatization, let
alone as an algebraic effect.

Hazel is a separation logic~\citep{VilhenaP21} for effect handlers
built on the concurrent separation logic Iris~\citep{JungKJBBD18} in
the Rocq proof assistant.
Hazel provides a powerful framework for reasoning about concurrency
effects implemented as effect handlers, but
it is quite a departure
from the elementary equational reasoning provided by the theory of
algebraic effects and
 gives little semantic insight into the effects
being defined.
In contrast, our work characterizes a particular concurrency effect
(dynamic threads) as an algebraic effect (specifically a parameterised
algebraic effect) corresponding to a natural denotational model.
Future work may adapt and extend our approach to support a broad
range of different concurrency effects or connect
to effect handlers and programming practice.

%

\subsection{Future Work}
The framework of algebraic theories allows for modular combination of effects~\cite{HPP06Combine}.
We could use this to combine concurrency based on dynamic threads with other effects such as global and local state~\cite{PlotkinP02} which is shared between threads, and to model probabilistic scheduling of threads.

We have used labelled posets (pomsets), which are standard in the study of true concurrency e.g.~\cite{DBLP:journals/ijpp/Pratt86}, as the notion of observation in our operational semantics.
We hope our denotational model can connect in the future with an operational semantics based on interleaving traces, which is more standard in process calculus e.g.~\cite{DBLP:books/daglib/0067019}.

Possible semantic variations of $\fork$ and $\wait$ abound, such as waiting for a thread and all its descendants to finish, or limiting the number of threads that can exist at one time.
Another extension involves threads that finish with a value rather than with the empty type, and so the $\wait$ operation returns that value to the parent. This extension is an abstract
form of inter-thread communication.
In this paper we concentrated on a minimal idealized fragment of POSIX-like threads, so there are many thread features that we could attempt to model in future work, such as allowing a thread to find out its own ID, or other thread synchronization mechanisms.

\subsection{Concluding Remarks}
We have studied the semantics of dynamic creation of threads using the framework of parameterized algebraic theories, by treating thread IDs as abstract parameters.
In~\Cref{sec:pres-theory-fork} we gave an algebraic theory that axiomatizes operations such as forking and waiting for threads.
In~\Cref{sec:rep-theorem} we provided a syntax-free characterization of terms in this theory~(\Cref{thm:main-theorem}) based on an extension of labelled posets, which are well-established in concurrency theory.
We then showed in~\Cref{sec:compl-theor-theory} that our theory is in a certain sense complete with respect to equality of ordinary labelled posets.

In~\Cref{sec:background}, we introduced a simple concurrent programming language and its operational semantics.
To model this language denotationally, in~\Cref{sec:denot-semant-progr}, we used our algebraic theory of dynamic threads and the connection between algebraic theories and monadic semantics.
We proved that the denotational semantics is adequate, sound and fully abstract at first order.

In summary, our simple language demonstrates that the theory of algebraic effects applies directly to concurrency primitives, and that it is profitable to pursue this algebraic perspective.

%
%
\begin{acks}
  Ohad Kammar, Jack Liell-Cock and Sam Staton were funded in part by a Royal Society University Research Fellowship, ERC Grant BLAST, AFOSR under award number FA9550-21-1-0038, a Clarendon Scholarship, and grants from ARIA Safeguarded AI.
  Sam Lindley and Cristina Matache were supported by UKRI Future Leaders Fellowship ``Effect Handler Oriented Programming'' (MR/T043830/1 and MR/Z000351/1) and by the Huawei Edinburgh Joint Lab project ``EPOCH: Effectful programming on capability hardware''.
For the purpose of open access, the authors have applied a Creative Commons Attribution (CC BY) licence to any Author Accepted Manuscript version arising from this submission.

We thank the anonymous referees and many people for helpful discussions and useful suggestions, including Rob van Glabbeek and Sean Moss.
\end{acks}

\bibliographystyle{ACM-Reference-Format}
\bibliography{refs.bib}

\end{document}